\documentclass[a4paper,onecolumn,11pt,accepted=2020-10-06]{quantumarticle}
\pdfoutput=1

\usepackage[utf8]{inputenc}
\usepackage[english]{babel}
\usepackage[T1]{fontenc}
\usepackage{hyperref}
\usepackage{url}
\usepackage{pgfplots}
\pgfplotsset{width=7cm,compat=1.9}
\usepackage{authblk}
\usepackage{amsmath,amsthm,amssymb, tabu} %
\usepackage{fullpage}
\usepackage{booktabs} %
\usepackage{dsfont} %
\usepackage{enumitem} %
\usepackage{float} %
\usepackage[margin=1in]{geometry} %
\usepackage{graphicx} %
\usepackage{titlesec}
\usepackage{mathtools}
\usepackage{systeme}
\usepackage{enumitem}
\usepackage{verbatim}
\usepackage{mdframed} %
\usepackage{microtype} %
\usepackage{soul} %
\usepackage{suffix} 
\usepackage{tikz} %
\usepackage{xcolor} %
\usepackage{xspace} %
\usepackage[scr=rsfso]{mathalfa} 
\usepackage{ragged2e}

\usepackage{qcircuit}

\definecolor{weborange}{rgb}{.8,.3,.3}
\definecolor{webblue}{rgb}{0,0,.8}
\definecolor{internallinkcolor}{rgb}{0,.5,0}
\definecolor{externallinkcolor}{rgb}{0,0,.5}

\usepackage[symbol]{footmisc} 

\usetikzlibrary{calc,positioning}

\newcommand{\qq}[1]{``#1''} 

\mdfdefinestyle{figstyle}{ %
  linecolor=black!7, %
  backgroundcolor=black!7, %
  innertopmargin=10pt, %
  innerleftmargin=25pt, %
  innerrightmargin=25pt, %
  innerbottommargin=10pt %
}

\definecolor{White}{rgb}{1,1,1} %
\definecolor{Black}{rgb}{0,0,0} %
\definecolor{LightGray}{rgb}{.8,.8,.8} %
\colorlet{ChannelColor}{LightGray} %
\colorlet{ChannelTextColor}{Black} %
\colorlet{ReadoutColor}{White} %

\numberwithin{equation}{section}
\newtheorem{thm}{Theorem}[section]
\newtheorem{lem}[thm]{Lemma}
\newtheorem*{lem*}{Lemma}
\newtheorem{cor}[thm]{Corollary}
\newtheorem{prop}[thm]{Proposition}
\newtheorem{question}[thm]{Question}
\theoremstyle{definition}
\newtheorem{defn}[thm]{Definition}

\theoremstyle{remark}
\newtheorem{rmk}[thm]{Remark}
\newtheorem{ex}[thm]{Example}

\theoremstyle{plain}
\newtheorem{innercustomgeneric}{\customgenericname}
\providecommand{\customgenericname}{}
\newcommand{\newcustomtheorem}[2]{%
  \newenvironment{#1}[1]
  {%
   \renewcommand\customgenericname{#2}%
   \renewcommand\theinnercustomgeneric{##1}%
   \innercustomgeneric
  }
  {\endinnercustomgeneric}
}

\newcustomtheorem{customthm}{Theorem}
\newcustomtheorem{customlemma}{Lemma}

\renewcommand{\hat}[1]{\widehat{#1}} 
\renewcommand{\tilde}[1]{\widetilde{#1}} 

\newcommand{\complex}{\mathbb{C}} %
\newcommand{\real}{\mathbb{R}} %
\newcommand{\integer}{\mathbb{Z}} %
\let\epsilon=\varepsilon %
\newcommand{\microspace}{\mspace{.5mu}} %
\newcommand{\ket}[1]{\ensuremath{\lvert\microspace #1
    \microspace\rangle}} %
\newcommand{\bra}[1]{\ensuremath{\langle\microspace #1
    \microspace\rvert}} %
\newcommand{\paren}[1]{(#1)}
\newcommand{\Paren}[1]{\left(#1\right)}
\newcommand{\bigparen}[1]{\big(#1\big)}

\newcommand{\abs}[1]{\lvert#1\rvert}

\newcommand{\norm}[1]{\lVert#1\rVert}
\newcommand{\Norm}[1]{\left\lVert#1\right\rVert}
\newcommand{\bignorm}[1]{\big\lVert#1\big\rVert}

\newcommand{\iprod}[1]{\langle#1\rangle}

\newcommand{\class}[1]{\textup{#1}\xspace} %

\newcommand{\QMIP}{\class{QMIP}} %
\WithSuffix\newcommand\QMIP*{\ensuremath{\class{QMIP}^*}} %
\newcommand{\MIP}{\class{MIP}} %
\WithSuffix\newcommand\MIP*{\ensuremath{\class{MIP}^*}} %

\newtheorem{theorem}{Theorem} 
\newtheorem{claim}[theorem]{Claim} %
\newcommand{\setft}[1]{\mathrm{#1}} %
\newcommand{\Proj}{\setft{Proj}} %
\newcommand{\Unitary}{\setft{U}} %
\newcommand{\Lin}{\setft{L}} %

\def\01{\{0,1\}}

\DeclareMathOperator{\Tr}{Tr}

\DeclareMathOperator{\CHSH}{CHSH}

\newcommand{\Mod}[1]{\ \mathrm{mod}\ #1}

\begin{document}

\title{A generalization of CHSH and the algebraic structure of optimal strategies}

\author{David Cui}
\email{dz.cui@mail.utoronto.ca}
\author{Arthur Mehta}
\email{arthur.mehta@mail.utoronto.ca}
\affiliation{Department of Mathematics, University of Toronto, Toronto, Canada.}
\author{Hamoon Mousavi}
\email{hmousavi@cs.toronto.edu}
\author{Seyed Sajjad Nezhadi}
\email{sajjad.nezhadi@mail.utoronto.ca}
\affiliation{Department of Computer Science, University of Toronto, Toronto, Canada.}

\maketitle
\begin{abstract}
\emph{Self-testing} has been a rich area of study in quantum information theory. It allows an experimenter to interact classically with a black box quantum system and to test that a specific entangled state was present and a specific set of measurements were performed. Recently, self-testing has been central to high-profile results in complexity theory as seen in the work on entangled games PCP of Natarajan and Vidick (FOCS 2018), iterated compression by Fitzsimons et al.\ (STOC 2019), and NEEXP in MIP* due to Natarajan and Wright (FOCS 2019). 
The most studied self-test is the CHSH game which features a bipartite system with two isolated devices. This game certifies the presence of a single EPR entangled state and the use of anti-commuting Pauli measurements. Most of the self-testing literature has focused on extending these results to self-test for tensor products of EPR states and tensor products of Pauli measurements.

In this work, we introduce an algebraic generalization of CHSH by viewing it as a linear constraint system (LCS) game, exhibiting self-testing properties that are qualitatively different. These provide the first example of LCS games that self-test non-Pauli operators resolving an open question posed by Coladangelo and Stark (QIP 2017). Our games also provide a self-test for states other than the maximally entangled state, and hence resolves the open question posed by Cleve and Mittal (ICALP 2012). Additionally, our games have $1$ bit question and $\log n$ bit answer lengths making them suitable candidates for complexity theoretic application.
This work is the first step towards a general theory of self-testing arbitrary groups. In order to obtain our results, we exploit connections between sum of squares proofs, non-commutative ring theory, and the Gowers-Hatami theorem from approximate representation theory. A crucial part of our analysis is to introduce a sum of squares framework that generalizes the \emph{solution group} of Cleve, Liu, and Slofstra (Journal of Mathematical Physics 2017) to the non-pseudo-telepathic regime. Finally, we give a game that is not a self-test by "gluing" together two copies of the magic square game. Our results suggest a richer landscape of self-testing phenomena than previously considered. 

\end{abstract}
\pagebreak
\section{Introduction}
In 1964, Bell showed that local hidden-variable theories, which are classical in nature, cannot explain all quantum mechanical phenomena \cite{Bell}. This is obtained by exhibiting a violation of a \emph{Bell inequality} by correlations arising from local measurements on an entangled state. Furthermore, in some instances, it is known that only certain measurements can produce these correlations. So through local measurements not only is it possible to verify that nature is not solely governed by classical theories, it is also possible to obtain conclusive statistical evidence that a specific quantum state was present and specific measurements were performed. Results of this nature are often referred to as \emph{self-testing} (also known as \emph{rigidity}), first formalized by Mayers and Yao in \cite{MaoSelfTest}. Self-testing has wide reaching applications in areas of theoretical computer science including complexity theory \cite{low-degree,iterated-compression,neexp}, certifiable randomness \cite{QuantumDice}, device independent quantum cryptography \cite{QuantumCrypto1,device-independent}, and delegated quantum computation \cite{verifier-on-a-leash}. See \cite{Selftesting} for a comprehensive review. Below we visit five natural questions on the topic of self-testing that we answer in this paper.  

The CHSH game \cite{CHSH} is the prototypical example of a \emph{non-local game}. In CHSH, two separated players, Alice and Bob, are each provided with a single classical bit, $s$ and $t$, respectively, chosen uniformly at random by a referee; the players reply with single classical bits $a$ and $b$ to the referee; and win the game if and only if $a \oplus b = s \wedge t$. Classically, the players can win the CHSH game with probability at most $75\%$. Remarkably, if we allow Alice and Bob to share an entangled state and employ a \emph{quantum strategy}, then the optimal winning probability is approximately $85\%$. For an introduction to non-local games, see \cite{Watrous}.

CHSH is also a canonical example of a self-testing game. Prior to the formalization of self-testing by Mayers and Yao it was already known \cite{Summers,SelfTestCHSH} that any optimal quantum strategy for CHSH must be, up to application of local isometries, using the Einstein-Podolsky-Rosen (EPR) state
\[ \ket{\psi}= \frac{1}{\sqrt{2}} \Paren{\ket{00} + \ket{11} }. \]

Self-testing can be framed either as an statement about non-local games, Bell inequalities, or more generally correlations. CHSH is an instance of a \emph{non-pseudo-telepathic} game. A \emph{pseudo-telepathic} game is one that exhibits \emph{quantum advantage} (i.e, its quantum value is strictly larger than that of its classical value) and its quantum value is $1$. CHSH can also be viewed as a \emph{linear constraint system} (LCS) game over $\mathbb{Z}_2$ \cite{BCSTensor}. 
LCS games are non-local games in which Alice and Bob cooperate to convince the referee that they have a solution to a system of linear equations. We introduce a new generalization of CHSH to a family of non-pseudo-telepathic LCS games over $\mathbb{Z}_{n}$ for all $n\geq 2$. These games resolve the following questions.

\begin{question}\label{Q1}
Are there states other than the maximally entangled state that can be self-tested by a non-local game?
\end{question}
To date much has been discovered about self-testing the maximally entangled state, $\frac{1}{\sqrt{d}}\sum^{d-1}_{j=0} \ket{j}\ket{j}$.
Mermin's \emph{magic square} game \cite{Mermin} can be used to self-test two copies of the EPR state and the \emph{parallel-repeated magic square} game can be used to self-test $2n$ copies of the EPR state \cite{RepeatedMagicSquare}.

The sum of squares (SOS) decomposition technique in \cite{TiltedCHSH} shows that the \emph{tilted CHSH} is a self-test for any pure state of two entangled qubits. This self-testing is stated in terms of violation of Bell inequalities. It is an open problem if the same applies for non-local games. The case for self-testing in higher dimensions has proven more difficult to analyze. Remarkably, it is still possible to self-test any bipartite entangled state, in any dimension \cite{selftestAnyState}. However, these self-test results are presented in terms of violations of correlations, unlike the CHSH game which arises from a non-local game (with binary payoff). Our games also resolve in the negative the question ``Can every LCS game be played optimally using the maximally entangled state?'' posed in \cite{BCSTensor}. 

\begin{question}\label{Q2}
Are there non-local games that provide a self-test for measurements that are not constructed from qubit Pauli operators?
\end{question}

The protocols in all of the above examples also provide a self-test for the measurement operators. That is if the players are playing optimally then they must, up to application of local isometries, have performed certain measurements. Self-testing proofs rely on first showing that operators in optimal strategies must satisfy certain algebraic relations. These relations help identify optimal operators as representations of some group. This is then used to determine the measurements and state up to local isometries. In the case of CHSH, one can verify that Alice and Bob's measurements must anti-commute if they are to play optimally. These relations are then enough to conclude that operators of optimal strategies generate the dihedral group of degree $4$ (i.e., the Pauli group). Thus CHSH is a self-test for the well-known Pauli matrices $\sigma_X$ and $\sigma_Z$ \cite{NewCHSH}. 

Self-tests for measurements in higher dimensions have been primarily focused on self-testing $n$-fold tensor-products of $\sigma_X$ and $\sigma_Z$ \cite{VidickNatarajan, ColadangeloParallel, Matthew}. It is natural to ask if there are self-tests for operators that are different than ones constructed from qubit Pauli operators. Self-testing Clifford observables has also been shown in \cite{Leash}. Our games provides another example that is neither Pauli nor Clifford. Since our games are LCS this resolves the question, first posed by \cite{RobustRigidityLCS}, in the affirmative. 

\begin{question}\label{Q3}
Can we extend the solution group formalism for pseudo-telepathic LCS games to a framework for proving self-testing for all LCS games?
\end{question}

The \emph{solution group} introduced in \cite{BCSCommuting} is an indispensable tool for studying pseudo-telepathic LCS games. To each such game there corresponds a group known as the solution group. Optimal strategies for these games are characterized by their solution group in the sense that any perfect quantum strategy must induce certain representations of this group. Additionally, the work in \cite{RobustRigidityLCS} takes this further by demonstrating a streamlined method to prove self-testing certain LCS games. It is natural to ask whether these methods can be extended to cover all LCS games. In this paper we make partial progress in answering this question by introducing a SOS framework, and use it to prove self-testing for our games. At its core, this framework utilizes the interplay between sum of squares proofs, non-commutative ring theory, and the Gowers-Hatami theorem \cite{GowersHatami,VidickBlog} from approximate representation theory.

\begin{question}\label{Q4}
Is there a systematic approach to design self-tests for arbitrary finite groups?
\end{question}

Informally a game is a self-test for a group if every optimal strategy induces a \emph{state dependent representation} of the group.
In every example that we are aware of, the self-tested solution group for pseudo-telepathic LCS games is the Pauli group. Slofstra, in \cite{embedding}, introduced an embedding theorem that embeds (almost) any finite group into the solution group of some LCS game. With the embedding theorem, the problem of designing games with certain properties reduces to finding groups with specific properties. Slofstra uses this connection to design games that exhibit separations between correlation sets resolving the `middle' Tsirelson's Problem.

However, there are three shortcomings to this approach. Firstly, the resulting game is very complex. Secondly, not all properties of the original group are necessarily preserved. Finally, the game is not a self-test for the original group. Our games self-test an infinite family of groups, non of which are the Paulis. One such example is the alternating group of degree $4$. The SOS framework makes partial progress towards a general theory for self-testing arbitrary groups.

\begin{question}\label{Q5}
Is there a non-local game that is not a self-test?
\end{question}
In addition to the infinite family of games, we introduce an LCS game that is obtained from \qq{gluing} together two copies of the magic square game. This \emph{glued magic square} provides an example of a game that is not a self-test \cite{Mermin}.  

\subsection{Main Results}\label{Results}
We introduce a family of non-local games $\mathcal{G}_{n}$ defined using the following system of equations over $\mathbb{Z}_n$
\begin{align*}
    &x_0x_1=1,  \\
    &x_0x_1=\omega_{n}.
\end{align*}
We are identifying $\mathbb{Z}_{n}$ as a multiplicative group and $\omega_n$ as the primitive $n$th root of unity. Note that the equations are inconsistent, but this does not prevent the game from being interesting. Alice and Bob try to convince a referee that they have a solution to this system of equations. Each player receives a single bit, specifying an equation for Alice and a variable for Bob, and subsequently each player returns a single number in $\mathbb{Z}_{n}$. Alice's response should be interpreted as an assignment to variable $x_0$ in the context of the equation she received, and Bob's response is interpreted as an assignment to the variable he received. The referee accepts their response iff their assignments are consistent and satisfy the corresponding equation. The case $n=2$ is the CHSH game. The classical value of these games is $\frac{3}{4}$. In Section \ref{LowerBounds}, we give a lower-bound on the \emph{quantum value} of this family of games. Specifically in Theorem \ref{LowerBoundThm}, we show that the quantum value is bounded below by
\[\frac{1}{2} + \frac{1}{2n\sin\Paren{ \frac{\pi}{2n} } } > \frac{3}{4}. \]
We show that the lower-bound is tight in the case of $n\leq 5$. We have numerical evidence that these lower-bounds are tight for all $n$. Specifically, we can find an upper-bound on the quantum value of a non-local game using the well-known hierarchy of semi-definite programs due to \cite{NPA}. It is of interest to note that the upper-bound is not obtained using the first level of the NPA hierarchy, as is the case with the CHSH game. Instead, the second level of this hierarchy was needed for $n\geq 3$. 

The optimal \emph{quantum strategy} for these games uses the entangled state 
\[ \ket{\psi_n} = \frac{1}{\gamma_n} \sum_{i=0}^{n-1}(1 - z^{n+2i+1})\ket{\sigma^i(0), \sigma^{-i}(0)} \in \mathcal{H}_A\otimes \mathcal{H}_B,\]
where $\gamma_n$ is the normalization factor, $\sigma_{n}=(0, 1, \dots, n-1)$ is a permutation, and $z_{n}$ is a $4n$'th root of unity. 
Observe that the state $\ket{\psi_n}$ has full Schmidt rank. Despite this, in all cases except $n=2$, the state $\ket{\psi_n}$ is not the maximally entangled state. For $n>2$, the entropy of our state is not maximal, but approaches the maximal entropy of $\log(n)$ in the limit. 

In Section \ref{Group generated by solutions}, we show that the group generated by the optimal strategy has the following presentation
\[ G_{n} = \left\langle P_{0}, P_{1}, J \ | \ P_{0}^{n}, P_{1}^{n}, J^{n}, [J, P_{0}], [J, P_{1}], J^{i}\Paren{ P_{0}^{i} P_{1}^{-i} }^{2} \text{ for } i = 1, 2, \dots, \left\lfloor n/2 \right\rfloor \right\rangle. \]
For example $G_3 = \integer_3\times A_4$ where $A_4$ is the alternating group of degree $4$. We show that our games are a self-test for these groups, for $n\leq 5$, in the sense that every optimal play of this game induces a representation of this group. We conjecture that this is true for all $n$. This partially resolves Question \ref{Q4}.

In section \ref{Upperbounds}, we analyze our game in the case $n=3$ and show that it can be used as a robust self-test for the following state
\[\frac{1}{\sqrt{10}} \Paren{ (1 - z^{4})\ket{00} + 2\ket{12} + (1 + z^{2})\ket{21} } \in \mathbb{C}^3 \otimes \mathbb{C}^3,\]
where $z:= e^{i \pi/6}$  is the primitive $12$th root of unity. Since this state is not the maximally entangled state, we have thus provided an answer to Question \ref{Q1}. This game also answers Question \ref{Q2} since it provides a robust self-test for the following operators
\begin{gather*}
    A_{0} = \begin{pmatrix} 
0 & 0 & 1 \\
1 & 0 & 0 \\
0 & 1 & 0
\end{pmatrix}, \quad A_{1} = \begin{pmatrix} 
0 & 0 & -z^2 \\
z^2 & 0 & 0 \\
0 & z^2 & 0
\end{pmatrix}, \\
B_{0} = \begin{pmatrix} 
0 & 0 & 1 \\
1 & 0 & 0 \\
0 & 1 & 0
\end{pmatrix}, \quad B_{1} = \begin{pmatrix} 
0 & -z^2 & 0 \\
0 & 0 & z^2 \\
z^2 & 0 & 0
\end{pmatrix},
\end{gather*}
which do not generate the Pauli group of dimension $3$.


In Section \ref{sec:sos_prelim}, we introduce the sum of squares framework, using an important lemma proven in Section \ref{sec:gh_prelim}, that gives a streamlined method for proving self-testing. We then use this framework to prove self-testing for our games. 
Furthermore, in Section \ref{BCSwithSOS}, we show that when restricted to pseudo-telepathic games, the SOS framework reduces to the solution group formalism of Cleve, Liu, and Slofstra \cite{BCSCommuting}.


In section \ref{semirigid}, we construct an LCS game that is obtained from ``gluing'' two copies of the magic square game together. This game is summarized in Figure \ref{fig:glued}. We exhibit two inequivalent perfect strategies and thus provide an answer to Question \ref{Q5}.

\begin{figure}[H]
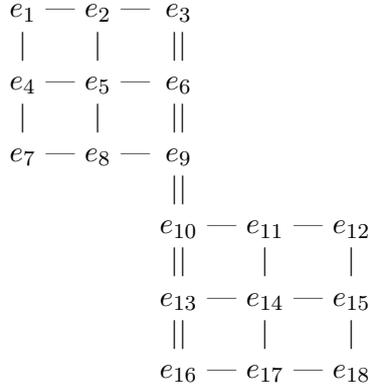

\[\begin{array}{c@{}c@{}c@{}c@{}c@{}c@{}c@{}c@{}c}
e_{1} & \text{ --- } & e_{2} & \text{ --- } & e_{3} & & & & \\
 \vert & & \vert & & \vert\vert& & & & \\
e_{4} & \text{ --- } & e_{5} & \text{ --- } & e_{6} & & & &\\
 \vert & & \vert & & \vert\vert & & & &\\
e_{7} & \text{ --- } & e_{8} & \text{ --- } & e_{9} & & & & \\

& & & & \vert\vert & & & \\

& & & & e_{10} & \text{ --- } & e_{11} & \text{ --- } & e_{12}\\
& & & & \vert\vert & & \vert & & \vert \\
& & & & e_{13} & \text{ --- } & e_{14} & \text{ --- } & e_{15}\\
& & & & \vert\vert & & \vert & & \vert \\
& & & & e_{16} & \text{ --- } & e_{17} & \text{ --- } & e_{18}\\
\end{array}\]

\caption{This describes an LCS game with 18 variables $e_{1}, e_{2}, \dots, e_{18}$. Each single-line indicates that the variables along the line multiply to $1$, and the double-line indicates that the variables along the line multiply to $-1$. }\label{fig:glued}
\end{figure}

\subsection{Proof techniques}
We prove self-testing in this paper following a recipe that we refer to as the \emph{SOS framework}. At its core it applies the Gowers-Hatami (GH) theorem which is a result in approximate-representation theory. GH has been used previously in proving self-testing, but some of the details have been overlooked in the literature. In this paper, we prove Lemma \ref{lemma:application_of_gh} that encapsulates the use of GH in proving self-testing. In Section \ref{sec:gh_prelim}, we define approximate representations, irreducible strategies, the Gowers-Hatami theorem and present the proof of the following lemma. 

\begin{lem*}[informal] Let $G_A,G_B$ be groups. Suppose every optimal strategy of the game $\mathcal{G}$ induces a pair of approximate representations of $G_A$ and $G_B$. Further suppose that there is a unique optimal irreducible strategy $(\rho,\sigma,\ket{\psi})$ where $\rho,\sigma$ are irreps of $G_A,G_B$, respectively. Then $\mathcal{G}$ is a self-test.\end{lem*}
Applying this lemma requires us to ascertain two properties of the game:
\begin{enumerate}
    \item Every optimal strategy induces approximate representations of some groups $G_A$ and $G_B$.
    \item There is a unique irreducible strategy $(\rho,\sigma,\ket{\psi})$ for the game $\mathcal{G}$.
\end{enumerate}

The first step is to obtain the bias expression for the game $\mathcal{G}$ that allows for a simple calculation of the wining probability of any startegy $\mathcal{S}=(\{A_i\},\{B_j\},\ket{\psi})$ (here $A_i$ and $B_j$ are Alice and Bob's measurement observables, respectively, and $\ket{\psi}$ is the shared state). The bias expression for $\mathcal{G}_n$ is given by
\[\mathcal{B}_n(A_0,A_1,B_0,B_1) = \sum_{i=1}^{n-1} A_{0}^{i} B_{0}^{-i}+ A_{0}^{i}  B_{1}^{i} + A_{1}^{i} B_{0}^{-i} + \omega^{-i} A_{1}^{i} B_{1}^{i}.\]
Then the winning probability of $\mathcal{S}$ is given by $\nu(\mathcal{G},\mathcal{S}) = \bra{\psi}\paren{\frac{1}{4n}\mathcal{B}_n(A_0,A_1,B_0,B_1) + \frac{1}{n}}\ket{\psi}$. For any real $\lambda$ for which there exist some polynomials $T_k$ giving a sum of squares decomposition such as
\[
\lambda I - \mathcal{B}_n(A_0,A_1,B_0,B_1) = \sum_{k} T_k^{\ast}(A_0,A_1,B_0,B_1) T_k(A_0,A_1,B_0,B_1),
\]
provides an upper bound of $\frac{\lambda}{4n}+\frac{1}{n}$ on the optimal value of the game (which we denote by $\nu^\ast(\mathcal{G}_n)$). This follows since expressing $\lambda I - \mathcal{B}_n$ as an SOS proves that it is a positive semidefinite operator and consequently $\bra{\psi}\mathcal{B}_n\ket{\psi} \leq \lambda $ for all states $\ket{\psi}$. 

Now if we have an SOS for $\lambda = 4n\nu^\ast(\mathcal{G})-4$,
then we can obtain some algebraic relations that every optimal strategy must satisfy. This follows since every optimal strategy must satisfy $\bra{\psi}\paren{\lambda I - B_n}\ket{\psi} = 0$, from which it follows $T_k \ket{\psi} = 0$ for all $k$.

Let $\paren{M_j(A_0,A_1) -I}\ket{\psi} = 0$ be all the relations derived from the SOS relations $T_k\ket{\psi}=0$ such that $M_i$ are monomials only in Alice's operators, and let $G_A$ be the group with the presentation 
\begin{align*}
G_A = \langle P_0, P_1 : M_i(P_0,P_1)\rangle
\end{align*}
We similarly obtain a group $G_B$ for Bob. These are the group referred in the above lemma. For the first assumption one must show that any optimal strategy gives approximate representations of these groups.

The next step is to prove the second assumption. We need to show that among all the pairs of irreps of $G_A$ and $G_B$ only one could give rise to an optimal strategy. To this end, we let $R_i(A_0,A_1)\ket{\psi} = 0$ be all the relations derived from relations $T_k\ket{\psi}=0$. These $R_i$ are allowed to be arbitrary polynomials (as opposed to monomials in the case of group relations). So any optimal irrep must satisfy all these polynomial relations. In some special cases, e.g., games $\mathcal{G}_n$, there is one polynomial relation that is enough to identify the optimal irreps.

\subsection{Relation to prior work}\label{Relation to other work}

Much work has been done to generalize CHSH to games over $\mathbb{Z}_n$. Initial generalizations were done by Bavarian and Shor \cite{Shor} and later extended in Kaniewski et al. \cite{WeakRigidity}. The game we present in section \ref{Zn Games} provides a different generalization by viewing CHSH as an LCS game. 
The classical value of our games is found to be $\frac{3}{4}$ from casual observation. Furthermore, we showcase quantum advantage by providing a lower bound on the quantum value for all $n$.

In contrast the generalization of CHSH discussed in Kaniewski et al. is so difficult to analyze that even the classical value is not known except in the cases of $n=3,5,7$. Additionally the quantum value of their Bell inequality is only determined after multiplying by choices of \qq{phase} coefficients. Self-testing for this generalization is examined by Kaniewski et al., where they prove self-testing for $n=3$ and show a weaker form of self-testing in the cases of $n=5,7$. For the games we introduce, we have self-testing for $n=3,4,5$ and we conjecture that they are self-tests, in the strict sense, for all $n$.

Furthermore, in \cite{Slofstra_2011}, Slofstra exhibits a game whose correlations are not extreme point, which suggests that it is also not a self-test, his result is not formulated in the language of self-testing and it would be interesting to rigorously show this to be the case. Independently of our work, in \cite{weak2}, a family of Bell inequalities, which includes the $I_{3322}$ game, is shown to self-test the maximally entangled state but no measurement operators. 

\subsection{Further work}\label{open-problems}
This paper leaves many open problems and avenues for further investigation. The most important of these follow.

\begin{enumerate}
    \item We conjecture that the class of games $\mathcal{G}_n$ are rigid for all $n$. The step missing from resolving this conjecture is an SOS decomposition $\nu(\mathcal{G}_n,\mathcal{S}_n)I - \mathcal{B}_n = \sum_{k}\alpha_{n,k} T_{n,k}^\ast T_{n,k}$ for $n > 5$ where polynomials $T_{n,k}$ viewed as vectors have unit norms and $\alpha_{n,k}$ are positive real numbers. 
    
    If this conjecture is true, then we have a simple family of games with $1$ bit question and $\log n$ bit answer sizes that are self-testing full-Schmidt rank entangled states of any dimension. In fact, we show that the amount of entanglement in these self-tested states rapidly approaches the maximum amount of entanglement. To the best of our knowledge this is the first example of a family of games with such parameters. 
    
    \item In Section \ref{Group generated by solutions}, we give efficient explicit presentations for $G_n$ and its multiplication table. Can we go further and characterize these groups in terms of direct and semidirect products of small well-known groups? The first few cases are as follows $$G_3 \cong \integer_3\times A_4, G_4\cong \paren{\integer_2^3\rtimes \integer_4}\rtimes \integer_4, G_5\cong \paren{\integer_2^4\rtimes \integer_5}\times \integer_5,$$ $$G_6 \cong \integer_3 \times \Paren{\paren{\paren{\paren{\integer_4 \times \integer_2^3} \rtimes \integer_2} \rtimes \integer_2} \rtimes \integer_3}.$$

    \item The third problem is to characterize all$\Mod n$ games over two variables and two equations. Let $(\integer_n,m_1,m_2)$ be the LCS game$\Mod n$ based on the system of equations 
    \begin{align*}
        x_0 x_1 &= \omega_n^{m_1}\\
        x_0 x_1 &= \omega_n^{m_2}.
    \end{align*}
    So for example $(\integer_n,0,1) = \mathcal{G}_n$. A full characterization includes explicit construction of optimal strategies, a proof of self-testing, and a characterization of the group generated by optimal strategies (i.e., the \emph{self-tested group}). 
    Interesting observations can be made about these games. For example $(\integer_4,0,2)$ self-tests the same strategy as $\CHSH$. Another interesting observation is that the self-tested group of $(\integer_3,0,1)$ and $(\integer_3,0,2)$ is $G_3 \cong \integer_3 \times A_4$, whereas the self-tested group of $(\integer_3,1,2)$ is $A_4$.  
    
    These games have similar bias expressions to those of $\mathcal{G}_n$. It is likely that the same kind of methods can be used to find optimal strategies and establish self-testing for these games. For example $(\integer_n,0,m)$ for all $m \in [n]\setminus \{0\}$ self-test the same group $G_n$. Just like $\mathcal{G}_n$, the representation theory of $G_n$ dictates the optimal strategies of all these games: the optimal irreducible strategies of $(\integer_n,0,m)$ for all $m \in [n]\setminus \{0\}$ are distinct irreps of $G_n$ of degree $n$.

    For example optimal strategies for all games $(\integer_5,0,m)$, where $m \in [5] \setminus \{0\}$, generate $G_5$. This group has $15$ irreps of degree five. For each $m \in [5]$, there are three irreps sending $J\to \omega_5^m I_5$. For each $m \in [5] \setminus \{0\}$, the unique optimal irrep strategy of $(\integer_5,0,m)$ is one of these three irreps. 
    
    These games are a rich source of examples for self-testing of groups. A full characterization is a major step toward resolving Question \ref{Q4}.
        
    \item One drawback of $\Mod n$ games is that the size of the self-tested groups grows exponentially, $\abs{G_n} = 2^{n-1}n^2$. Where are the games that self-test smaller groups for example the dihedral group of degree $5$, $D_5$? It seems that to test more groups, we need to widen our search space. 
    
    In a similar fashion to$\Mod n$ games, define games $(G,g_1,g_2)$ where $G$ is a finite group and $g_1,g_2 \in G$, based on the system of equations
    \begin{align*}
        x_0 x_1 &= g_1\\
        x_0 x_1 &= g_2.
    \end{align*}
    Understanding the map that sends $(G,g_1,g_2)$ to the self-tested group helps us develop a richer landscape of group self-testing. 
    
    \item How far can the SOS framework be pushed to prove self-testing? The first step in answering this question is perhaps a characterization of games $(G,g_1,g_2)$ (and their variants, e.g., system of equations with more variables and equations) using this framework.
    
    \item Glued magic square, as presented in Section \ref{semirigid}, is not a self-test for any operator solution, but both inequivalent strategies that we present use the maximally entangled state. Is the glued magic square a self-test for the maximally entangled state? If true, this would give another example of a non-local game that only self-tests the state and not the measurement operators.  
    
    After the publication of our work, Man\v{c}inska et al. \cite{Mancinska2} showed that this is indeed the case; specifically they showed that the glued magic square self-tests convex combinations of the two inequivalent strategies we presented in our work. Along with \cite{weak2}, these positively resolve a question asked in \cite{Selftesting} in the context of non-local games.
    
\end{enumerate}

\subsection{Organization of paper}
In section \ref{Prelims}, we fix the nomenclature and give basic definitions for non-local games, winning strategies, self-testing, LCS games, approximate representation, and the Gowers-Hatami theorem. In section \ref{Zn Games}, we give the generalization of CHSH and derive the bias operator of these games, that is used in the rest of the paper. In Section \ref{LowerBounds}, we establish lower-bounds on the quantum value for these games by presenting explicit strategies. In this section we also analyse the entanglement entropy of the shared states in these explicit strategies. In Section \ref{Group generated by solutions}, we give a presentation for the groups generated by Alice and Bob's observables. In Section \ref{sec:sos_prelim}, we present the SOS framework and give a basic example of its application in proving self-testing. In section \ref{Upperbounds}, we use the SOS framework to show that our lower-bound is tight in the case of $n=3$, and answer the questions we posed about self-testing. In section \ref{BCSwithSOS}, we show that the SOS framework reduces to the solution group formalism in the case of pseudo-telepathic LCS games. Finally, in Section \ref{semirigid} we provide an example of a non-rigid game.

\section*{Acknowledgements}
We would like to thank Henry Yuen for having asked many of the motivating questions for this project, as well as many insightful discussions and helpful notes on writing this paper. We also thank the anonymous reviewers for their suggestions and for suggesting the name glued magic square.

\tableofcontents
\section{Preliminaries}\label{Prelims}
We assume the reader has a working understanding of basic concepts from the field of quantum information theory. For an overview of quantum information, refer to \cite{TQI, QCQI,PaulsenNotes}.

\subsection{Notation}

We use $G$ to refer to a group, while $\mathcal{G}$ is reserved for a non-local game. Let $[n, m]$ denote the set $\{n, n+1, \dots, m\}$ for integers $n \leq m$, and the shorthand $[n] = [0, n-1]$. This should not be confused with $[X,Y]$, which is used to denote the commutator $XY - YX$. We let $I_n$ denote the $n \times n$ identity matrix and $e_i$, for $i\in [n]$, be the ith standard basis vector. The pauli observables are denoted $\sigma_x, \sigma_y,$ and $\sigma_z$. The Kronecker delta is denoted by $\delta_{i,j}$.

We will let $\mathcal{H}$ denote a finite dimensional Hilbert space and use the notation $\ket{\psi} \in \mathcal{H}$ to refer to vectors in  $\mathcal{H}$. We use $\Lin(\mathcal{H})$ to denote the set of linear operators in the Hilbert space $\mathcal{H}$. We use $\Unitary_n(\complex)$ to denote the set of unitary operators acting on the Hilbert space $\complex^n$. The set of projection operators acting on $\mathcal{H}$ are denoted by $\Proj(\mathcal{H})$. Given a linear operator $A \in \Lin(\mathcal{H})$, we let $A^* \in \Lin(\mathcal{H})$ denote the adjoint operator. For $X,Y \in \Lin(\mathcal{H})$, the Hilber-Schmidt inner product is given by $\iprod{X,Y}=\Tr(X^\ast Y)$. We also use the following shorthands $\Tr_{\rho}(X)= \Tr(X \rho)$ and $\iprod{X,Y}_\rho = \Tr_{\rho}(X^\ast Y)$  where $X,Y\in \Lin(\mathcal{H})$ and $\rho$ is a density operator acting on $\mathcal{H}$ (i.e., positive semidefinite with trace $1$). The von Neumann entropy of a density matrix $\rho$ is given by $S(\rho) = -\Tr(\rho \log{\rho})$.

We use $\Re(\alpha)$ to denote the real part of a complex number $\alpha$. We let $\omega_n = e^{2i\pi/n}$ be the $n$th root of unity. The Dirichlet kernel is $\mathcal{D}_{m}(x) = \frac{1}{2\pi}\sum_{k = -m}^{m} e^{ikx}$ which by a well known identity is equal to $ \frac{ \sin\Paren{ \Paren{ m + \frac{1}{2} } x } }{2\pi \sin\Paren{ \frac{x}{2} } }$.

The maximally entangled state with local dimension $n$ is given by $\ket{\Phi_n} = \frac{1}{\sqrt{n}} \sum_{i=0}^{n-1} \ket{i}\ket{i} \in \complex^n\otimes \complex^n$.

Let $\mathcal{H}_A,\mathcal{H}_B$ be Hilbert spaces of dimension $n$ and $\ket{\psi} \in \mathcal{H}_A \otimes \mathcal{H}_B$ be a bipartite state. Then there exists orthonormal bases $\{\ket{i_A}\}_{i=0}^{n-1}$ for $\mathcal{H}_A$ and $\{\ket{i_B}\}_{i=0}^{n-1}$ for $\mathcal{H}_B$ and unique non-negative real numbers $\{\lambda_i\}_{i=0}^{n-1}$ such that $\ket{\psi} = \sum_{i=0}^{n-1} \lambda_i \ket{i_A} \ket{i_B}$. The $\lambda_i$'s are known as Schmidt coefficients.  

The Schmidt rank of a state is the number of non-zero Schmidt coefficients $\lambda_i$. The Schmidt rank is a rough measure of entanglement. In particular, a pure state $\ket{\psi}$ is entangled if and only if it has Schmidt rank greater than one.

Another measure of entanglement is the \emph{entanglement entropy}. Given the Schmidt decomposition of a state $\ket{\psi} = \sum_{i=0}^{n-1} \lambda_i \ket{i_A} \ket{i_B}$, the entanglement entropy $S_{\psi}$ is given by $- \sum_{i=0}^{n-1} \lambda_i^2 \log(\lambda_i^2)$. The maximum entanglement entropy is $\log(n)$. A pure state is separable (i.e. not entangled) when the entanglement entropy is zero. If the entanglement entropy of a state $\ket{\psi}$ is maximum, then the state is the maximally entangled state up to local unitaries, i.e., there exist unitaries $U_A,U_B \in \Unitary_n(\complex)$, such that $\ket{\psi} = U_A \otimes U_B \ket{\Phi_n}$.

\subsection{Non-local games}\label{sec:non-local-games}
A \emph{non-local game} is played between a referee and two cooperating players Alice and Bob who cannot communicate once the game starts. The referee provides each player with a question (input), and the players each respond with an answer (output). The referee determines whether the players win with respect to fixed conditions known to all parties. Alice does not know Bob's question and vice-versa as they are not allowed to communicate once the game starts. However, before the game starts, the players could agree upon a strategy that maximizes their success probability. Below we present the formal definition and some accompanying concepts.

\begin{defn}
A non-local game $\mathcal{G}$ is a tuple $(\mathcal{I}_A, \mathcal{I}_B, \mathcal{O}_A, \mathcal{O}_B, \pi, V) $ where $\mathcal{I}_A$ and $\mathcal{I}_B$ are finite question sets, $\mathcal{O}_A$ and $\mathcal{O}_B$ are finite answer sets, $\pi$ denotes the probability distribution on the set $\mathcal{I}_A \times \mathcal{I}_B$  and $V:\mathcal{I}_A \times \mathcal{I}_B \times \mathcal{O}_A \times \mathcal{O}_B \rightarrow \lbrace 0 , 1 \rbrace$ defines the winning conditions of the game. 

When the game begins, the referee chooses a pair $(i,j)\in \mathcal{I}_A \times \mathcal{I}_B$ according to the distribution $\pi$. The referee sends $i$ to Alice and $j$ to Bob. Alice then responds with $a\in \mathcal{O}_A$ and Bob with $b \in \mathcal{O}_B$. The players win if and only if $V(i,j,a,b) = 1$.    
\end{defn}

A \emph{classical strategy} is defined by a pair of functions $f_A: \mathcal{I}_A \rightarrow \mathcal{O}_A$ for Alice and $f_B: \mathcal{I}_B \rightarrow \mathcal{O}_B$ for Bob. The winning probability of this strategy is $$\sum_{i,j} \pi(i,j)V(i,j,f_A(i),f_B(j)).$$ The \emph{classical value}, $\nu(\mathcal{G})$, of a game is the supremum of this quantity over all classical strategies $(f_A,f_B)$. 

A \emph{quantum strategy} $\mathcal{S}$ for $\mathcal{G}$ is given by Hilbert spaces $\mathcal{H}_A$, $\mathcal{H}_B$, a state $\ket{\psi} \in \mathcal{H}_A \otimes \mathcal{H}_B$, and projective measurements $\lbrace E_{i,a} \rbrace_{a\in \mathcal{O}_A} \subset \Proj(\mathcal{H}_A)$ and $\lbrace F_{j,b} \rbrace_{b\in \mathcal{O}_B} \subset \Proj(\mathcal{H}_B)$ for all $i\in \mathcal{I}_A$ and $j \in \mathcal{I}_B$. 

Alice and Bob each have access to Hilbert spaces $\mathcal{H}_A$ and $\mathcal{H}_B$ respectively. On input $(i,j)$, Alice and Bob measure their share of the state $\ket{\psi}$ according to $\lbrace E_{i,a} \rbrace_{a\in \mathcal{O}_A}$ and $\lbrace F_{j,b} \rbrace_{b\in \mathcal{O}_B}$. The probability of obtaining outcome $a,b$ is given by
$\bra{\psi} E_{i,a} \otimes F_{j,b} \ket{\psi}.$ The winning probability of strategy $\mathcal{S}$, denoted by $\nu(\mathcal{G},\mathcal{S})$ is therefore $$\nu(\mathcal{G},\mathcal{S}) = \sum_{i,j,a,b} \pi(i,j)\bra{\psi} E_{i,a} \otimes F_{j,b} \ket{\psi}V(i,j,a,b).$$ The quantum value of a game, written $\nu^*( \mathcal{G})$, is the supremum of the winning probability over all quantum strategies.

The famous CHSH game \cite{CHSH} is the tuple $(\mathcal{I}_A, \mathcal{I}_B, \mathcal{O}_A, \mathcal{O}_B, \pi, V)$ where $\mathcal{I}_A = \mathcal{I}_B = \mathcal{O}_A=\mathcal{O}_B = \{0,1\}$, $\pi$ is the uniform distribution on $\mathcal{I}_A\times \mathcal{I}_B$, and $V(i,j,a,b)=1$ if and only if
\[ a +b \equiv i j \mod 2 .\] 
The CHSH game has a classical value of $0.75$ and a quantum value of $\frac{1}{2}+\frac{\sqrt{2}}{4} \approx 0.85$ \cite{CHSH}.

A strategy $\mathcal{S}$ is optimal if $\nu(\mathcal{G}, \mathcal{S}) = \nu^{*}(\mathcal{G})$. When a game's quantum value is larger than the classical value we say that the game exhibits \emph{quantum advantage}. A game is \emph{pseudo-telepathic} if it exhibits quantum advantage and its quantum value is $1$.

An \emph{order-$n$ generalized observable} is a unitary $U$ for which $U^n = I$. It is customary to assign an order-$n$ generalized observable to a projective measurement system $\lbrace E_0, \dots , E_{n-1} \rbrace$ as
\[ A= \sum^{n-1}_{i=0} \omega_n^{i}E_i. \]
Conversely, if $A$ is an order-$n$ generalized observable, then we can recover a projective measurement system $\{ E_0, \dots, E_{n-1} \}$ where
\[ E_i = \frac{1}{n}\sum^{n-1}_{k=0} \left( \omega_n^{-i} A\right)^k. \]
In this paper, present strategies in terms of generalized observables.

Consider the strategy $\mathcal{S}$ consisting of the shared state $\ket{\psi}\in \mathcal{H}_A\otimes \mathcal{H}_B$ and observables $\{A_i\}_{i\in \mathcal{I}_A}$ and $\{B_j\}_{j\in \mathcal{I}_B}$ for Alice and Bob. We say the game $\mathcal{G}$ is a \emph{self-test} for the strategy $\mathcal{S}$ if there exist $\epsilon_0 \geq 0$ and $\delta: \real^+\to\real^+$  a continuous function with $\delta(0) = 0$, such that the following hold
\begin{enumerate}
    \item $\mathcal{S}$ is optimal for $\mathcal{G}$.
    \item For any $0 \leq \epsilon \leq \epsilon_0$ and any strategy $\tilde{\mathcal{S}} = (\{\tilde{A}_i\}_{i\in \mathcal{I}_A},\{\tilde{B_j}\}_{j\in \mathcal{I}_B},\ket{\tilde{\psi}})$ where $\ket{\tilde{\psi}}\in \tilde{\mathcal{H}}_A\otimes \tilde{\mathcal{H}}_B$ and $\nu(\mathcal{G},\tilde{\mathcal{S}}) \geq \nu^\ast(\mathcal{G}) - \epsilon$, there exist local isometries $V_A$ and $V_B$, and a state $\ket{\text{junk}}$ such that the following hold
    \begin{itemize}
        \item $\bignorm{V_A\otimes V_B \ket{\tilde{\psi}} - \ket{\psi}\ket{\text{junk}}} \leq \delta(\epsilon)$,
        \item $\bignorm{V_A \tilde{A}_i \otimes V_B \ket{\tilde{\psi}} - \paren{A_i\otimes I \ket{\psi}}\ket{\text{junk}}} \leq \delta(\epsilon)$ for all $i \in \mathcal{I}_A$,
        \item $\bignorm{V_A\otimes V_B \tilde{B_j} \ket{\tilde{\psi}} - \paren{I\otimes B_j\ket{\psi}}\ket{\text{junk}}} \leq \delta(\epsilon)$ for all $j \in \mathcal{I}_B$.
    \end{itemize}
\end{enumerate}
We use the terminology \emph{rigidity} and self-testing interchangeably. \emph{Exact rigidity} is a weaker notion in which, we only require the second condition to hold for $\epsilon = 0$. In Section \ref{sec:sos_prelim}, we give as an example the proof of exact rigidity of the CHSH game. 

\subsection{Linear constraint system games}\label{LCS Prelims}
A \emph{linear constraint system} (LCS) game is a non-local game in which Alice and Bob cooperate to convince the referee that they have a solution to a system of linear equations over $\mathbb{Z}_n$. The referee sends Alice an equation and Bob a variable in that equation, uniformly at random. In response, Alice specifies an assignment to the variables in her equation and Bob specifies an assignment to his variable. The players win exactly when Alice's assignment satisfies her equation and Bob's assignment agrees with Alice. It follows that an LCS game has a perfect classical strategy if and only if the system of equations has a solution over $\integer_n$. Similarly the game has a perfect quantum strategy if and only if the system of equations, when viewed in the multiplicative form, has an \emph{operator solution} \cite{BCSTensor}. 

To each LCS game there corresponds a group referred to as the \emph{solution group}. The representation theory of solution group is an indispensable tool in studying pseudo-telepathic LCS games \cite{BCSCommuting,RobustRigidityLCS}. In what follows we define these terms formally, but the interested reader is encouraged to consult the references to appreciate the motivations. In this paper, we are interested in extending solution group formalism to general LCS games using the sum of squares approach. We explore this extension in Section \ref{Upperbounds}. When restriced to psuedo-telepathic LCS games, our SOS approach is identical to the solution group formalism. We present this in section \ref{BCSwithSOS} for completeness. 

Consider a system of linear equations $Ax=b$ where $A \in \mathbb{Z}_n^{r\times s}$, $b \in \mathbb{Z}_n^r$. We let $V_i$ denote the set of variables occurring in equation $i$ 
\[V_i = \lbrace j \in [s] : a_{i,j} \neq 0 \rbrace . \]
To view this system of linear equations in multiplicative form, we identify $\mathbb{Z}_n$ multiplicatively as $\{1,\omega_n,\ldots,\omega_n^{n-1}\}$. Then express the $i$th equation as
\[ \prod_{j \in V_{i}} x^{a_{ij}}_j = \omega_n^{b_{i}}. \]
In this paper we only use this multiplicative form. We let $S_i$ denote the set of satisfying assignments to equation $i$. 
In the LCS game $\mathcal{G}_{A,b}$, Alice receives an equation $i\in [r]$ and Bob receives a variable $j \in V_i$, uniformly at random. Alice responds with an assignment $x$ to variables in $V_i$ and Bob with an assignment $y$ to his variable $j$. They win if $x\in S_i$ and $x_j=y$.

The solution group $G_{A,b}$ associated with $\mathcal{G}_{A,b}$, is the group generated by $g_1, \dots, g_s, J,$ satisfying the relations
\begin{enumerate}
    \item $g^n_j=J^n=1$ for all $j$,
    \item $g_jJ=Jg_j$ for all $j$,
    \item $g_jg_k=g_kg_j$ for $j,k \in V_i$ for all $i$, and
    \item $\prod_{j \in V_i}g^{A_{ij}}_j = J^{b_{i}}$.
\end{enumerate}

\subsection{Gowers-Hatami theorem and its application to self-testing}\label{sec:gh_prelim}

In order to precisely state our results about self-testing in Section \ref{Upperbounds}, we recall the Gowers-Hatami theorem and $(\epsilon, \ket{\psi})$-representation \cite{GowersHatami,RobustRigidityLCS,VidickBlog}.
\begin{defn}\label{Defnrepresentation}
Let $G$ be a finite group, $n$ an integer, Hilbert spaces $\mathcal{H}_A,\mathcal{H}_B$ of dimension $n$, and $\ket{\psi}\in \mathcal{H}_A\otimes \mathcal{H}_B$ a state with the reduced density matrix $\sigma \in \Lin(\mathcal{H}_A)$. An $ (\epsilon, \ket{\psi})$-representation of $G$, for $\epsilon \geq 0$, is a function $f: G \rightarrow U_n(\complex)$ such that  
\begin{align}
 \mathbb{E}_{x,y} \Re \left( \iprod{f(x)^\ast f(y),f(x^{-1}y)}_{\sigma} \right)\geq 1-\epsilon.\label{representation_condition}
\end{align}
 
\end{defn}
In the case of $\epsilon =0$, we abbreviate and call such a map a $\ket{\psi}$-representation, in which case the condition \ref{representation_condition} simplifies to
\begin{align*}
 \iprod{f(x)^\ast f(y),f(x^{-1}y)}_{\sigma}=1,
\end{align*}
or equivalently 
\begin{align}
 f(y)^*f(x)f(x^{-1}y)\ket{\psi}= \ket{\psi},\label{simplified_representation_condition}
\end{align}
for all $x,y\in G$. In Condition (\ref{simplified_representation_condition}), we are implicitly dropping the tensor with identity on $\mathcal{H}_B$. Note that a $\ket{\psi}$-representation $f$ is just a group representation 
when restricted to the Hilbert space $\mathcal{H}_0 = \text{span} \lbrace f(g) \ket{\psi}: g \in G \rbrace$, i.e., the Hilbert space generated by the image of $f$ acting on $\ket{\psi}$. To see this, we first rewrite (\ref{simplified_representation_condition}) as \[ f(x^{-1}y)\ket{\psi} = f(x)^*f(y)\ket{\psi}.\] Thus for any $x,y \in G$ we have 
\begin{align*}  f(x^{-1})^*f(x^{-1}y)\ket{\psi} =f(xx^{-1}y)\ket{\psi}= f(y)\ket{\psi}. \end{align*}
We can multiply both sides by $f(x^{-1})$ to obtain $ f(x^{-1}y)\ket{\psi}= f(x^{-1})f(y)\ket{\psi}$ for all $x,y \in G$ or equivalently
\begin{align}\label{simplified_representation_condition_2}
    f(x)f(y)\ket{\psi} = f(xy)\ket{\psi} \text{ for all } x,y \in G.
\end{align}
This shows that for all $x \in G$, the operator $f(x)$ leaves the subspace $H_0$ invariant. Thus we can view $f(x)|_{H_{0}}$, the restriction of $f(x)$ to this subspace, as an element of $\Lin(H_0)$. Furthermore, by (\ref{simplified_representation_condition_2}), the map $x \mapsto f(x)|_{H_{0}}$ is a homormorphism and thus a representation of $G$ on $H_{0}$. 

We need the following special case of the Gowers-Hatami (GH) theorem as presented in \cite{VidickBlog}. The analysis of the robust rigidity of these games uses the general statement of GH, using $(\epsilon,\ket{\psi})$-representation. Although skipped in this paper, the tools are in place to analyse the robust case.

\begin{thm}[Gowers-Hatami]\label{gh_theorem} Let $d$ be an integer, $\ket{\psi} \in \complex^d\otimes \complex^d$ a bipartite state, $G$ a finite group, and $f:G\rightarrow \Unitary_d(\complex)$ a $\ket{\psi}$-representation. Then there exist $d'\geq d$, a representation $g: G\rightarrow \Unitary_{d'}(\complex)$, and an isometry $V:\complex^d\rightarrow \complex^{d'}$ such that $f(x)\otimes I\ket{\psi} = V^\ast g(x) V\otimes I\ket{\psi}$.
\end{thm}
From the proof of this theorem in \cite{VidickBlog}, we can take $g = \oplus_{\rho} I_d\otimes I_{d_\rho}\otimes \rho$ where $\rho$ ranges over irreducible representations of $G$ and $d_\rho$ is the dimension of $\rho$. Additionally, in the same bases, we can factorize $V$ into a direct sum over irreps such that $Vu = \oplus_{\rho} \paren{V_\rho u}, \text{ for all } u\in \complex^d$ where $V_{\rho} \in \Lin(\complex^d,\complex^d\otimes \complex^{d_\rho}\otimes \complex^{d_\rho})$ are some linear operators. It holds that $ \sum_\rho V_\rho^\ast V_\rho  = V^\ast V = I_{d}$. 

In some special cases, such as in our paper, we can restrict $g$ to be a single irreducible representation of $G$. In such cases we have a streamlined proof of self-testing. Lemma \ref{lemma:application_of_gh} below captures how GH is applied in proving self-testing in these cases.

Let $\mathcal{G}=(\mathcal{I}_A, \mathcal{I}_B, \mathcal{O}_A, \mathcal{O}_B, \pi, V)$ be a game, $G_A$ and $G_B$ be groups with generators $\{P_i\}_{i\in I_A}$ and $\{Q_j\}_{j\in I_B}$, $\hat{G}_A$ and $\hat{G}_B$ be free groups over $\{P_i\}_{i\in I_A}$ and $\{Q_j\}_{j\in I_B}$, and $\mathcal{S} = (\{A_i\},\{B_j\},\ket{\psi})$ be a strategy where $\ket{\psi} \in \complex^{d_A}\otimes \complex^{d_B}$. We define two functions $f_A^\mathcal{S} : \hat{G}_A \to \Unitary_{d_A}(\complex)$, $f_B^\mathcal{S}: \hat{G}_B \to \Unitary_{d_B}(\complex)$ where $f_A^{\mathcal{S}}(P_i) = A_i$ and $f_B^{\mathcal{S}}(Q_j) = B_j$ and they are extended homomorphically to all of $\hat{G}_A$ and $\hat{G}_B$, respectively. Suppose that the game $\mathcal{G}$ has the property that for every optimal strategy $\tilde{\mathcal{S}} = (\{\tilde{A}_i\},\{\tilde{B}_j\},\ket{\tilde{\psi}})$,  $f_A^{\tilde{\mathcal{S}}}$ and $f_B^{\tilde{\mathcal{S}}}$ are $\ket{\tilde{\psi}}$-representations for $G_A$ and $G_B$, respectively. 
 
 Now applying GH, for every optimal strategy $\tilde{\mathcal{S}}$, there exist representations $g_A, g_B$ of $G_A, G_B$, respectively, and isometries $V_A,V_B$ such that 
 \begin{align*}
 f_A^{\tilde{\mathcal{S}}}(x)\otimes I \ket{\tilde{\psi}} = V_A^\ast g_A(x) V_A \otimes I\ket{\tilde{\psi}} \text{ for all } x \in G_A,\\
 I \otimes f_B^{\tilde{\mathcal{S}}}(y)\ket{\tilde{\psi}} = I \otimes V_B^\ast g_B(y) V_B \ket{\tilde{\psi}}\text { for all } y \in G_B.
 \end{align*}
 Unfortunately this is not enough to establish rigidity for $\mathcal{G}$ as defined in Section \ref{sec:non-local-games}. To do this, we need and extra assumption on $\mathcal{G}$ that we deal with in the following lemma. 

For any pair of representations $\rho,\sigma$ of $G_A,G_B$ respectively, and state $\ket{\psi} \in \complex^{d_\sigma}\otimes \complex^{d_\rho}$, let $\mathcal{S}_{\rho,\sigma,\ket{\psi}} = (\{\rho(P_i)\}_{i\in \mathcal{I}_A},\{\sigma(Q_j)\}_{j\in \mathcal{I}_B},\ket{\psi})$ be the strategy induced by the pair of representations $(\rho,\sigma)$. Also let $\nu(\mathcal{G},\rho,\sigma) = \max_{\ket{\psi}} \nu(\mathcal{G},\mathcal{S}_{\rho,\sigma,\ket{\psi}})$. 

\begin{lem}\label{lemma:application_of_gh}
Suppose that there is only one pair of irreps $\bar{\rho},\bar{\sigma}$ for which $\nu(\mathcal{G},\bar{\rho},\bar{\sigma})=\nu^\ast(\mathcal{G})$. Additionally assume that $\ket{\psi}$ is the unique state (up to global phase) for which $\mathcal{S}_{\bar{\rho},\bar{\sigma},\ket{\psi}}$ is an optimal strategy. Let $\tilde{\mathcal{S}} = (\{\tilde{A}_i\},\{\tilde{B}_j\},\ket{\tilde{\psi}})$ be an optimal strategy of $\mathcal{G}$ such that $\ket{\tilde{\psi}} \in \complex^{d_A}\otimes \complex^{d_B}$, $f_A^{\tilde{\mathcal{S}}}$ and $f_B^{\tilde{\mathcal{S}}}$ are $\ket{\tilde{\psi}}$-representations for $G_A$ and $G_B$, respectively.
Then there exist isometries $V_A: \complex^{d_A}\to \complex^{d_A \abs{G_A}},V_B:\complex^{d_B}\to\complex^{d_B \abs{G_B}}$, and a state $\ket{\text{junk}}$ such that 
\begin{align*}
V_A\otimes V_B \ket{\tilde{\psi}} &= \ket{\text{junk}}\ket{\psi},\\
V_A \tilde{A}_i\otimes V_B \ket{\tilde{\psi}} &= \ket{\text{junk}} \bar{\rho}(P_i) \otimes I_{d_{\bar{\sigma}}} \ket{\psi},\\
V_A \otimes V_B\tilde{B}_j\ket{\tilde{\psi}} &= \ket{\text{junk}} I_{d_{\bar{\rho}}} \otimes \bar{\sigma}(Q_j) \ket{\psi},
\end{align*} 
for all $i\in I_A, j\in I_B$. 
\end{lem}
\begin{proof}
For simplicity, we only prove the case of binary games, i.e., we assume $\abs{\mathcal{O}_A}=\abs{\mathcal{O}_B}=2$. The general case follows similarly. For binary games we only need to consider strategies comprised of binary observables ($A$ is a binary observable if it is Hermitian and $A^2 = I$). Without loss of generality, we can assume that there exist some complex numbers $\lambda_{ij},\lambda_i,\lambda_j,\lambda$ such that for any strategy $\mathcal{S} = (\{A_i\},\{B_j\},\ket{\psi})$  
\begin{align}\label{eq:winning probability in bias form}
\nu(\mathcal{G},\mathcal{S}) = \bra{\psi}\Paren{\sum_{i\in I_A,j \in I_B} \lambda_{ij}A_i \otimes B_j +\sum_{i\in I_A} \lambda_{i}A_i \otimes I+ \sum_{j \in I_B} \lambda_{j}I \otimes B_j + \lambda I\otimes I}\ket{\psi}.
\end{align}

As argued earlier, by GH, we have
\begin{align}f_A^{\tilde{\mathcal{S}}}(x)\otimes I \ket{\tilde{\psi}} &= V_A^\ast g_A(x) V_A \otimes I\ket{\tilde{\psi}},\label{eq:first consequence of gh 1}\\
I \otimes f_B^{\tilde{\mathcal{S}}}(x) \ket{\tilde{\psi}} &= I\otimes V_B^\ast g_B(x) V_B\ket{\tilde{\psi}}\label{eq:first consequence of gh 2},
\end{align}
where $g_A = \oplus_{\rho} I_{d_A d_\rho}\otimes \rho,g_B = \oplus_{\sigma} I_{d_B d_\sigma}\otimes \sigma,$ where $\rho$ and $\sigma$ range over irreducible representations of $G_A$ and $G_B$, respectively. We also have the factorization $V_A u = \oplus_{\rho} \paren{V_{A,\rho} u}, \text{ for all } u\in \complex^{d_A}$ as well as $V_B u = \oplus_{\sigma} \paren{V_{B,\sigma} u}, \text{ for all } u\in \complex^{d_B}$. As mentioned above in the discussion that followed Theorem \ref{gh_theorem}, $V_{A,\rho}$ and $V_{B,\sigma}$ are some linear operators for which $\sum_{\rho} V_{A,\rho}^\ast V_{A,\rho} = I_{d_A}$ and $\sum_{\sigma}V_{B,\sigma}^\ast V_{B,\sigma} = I_{d_B}$. 

We want to write the winning probability of $\tilde{\mathcal{S}}$ in terms of the winning probabilities of irrep strategies. To this end, let 
\begin{align*}
p_{\rho,\sigma} &= \norm{V_{A,\rho}\otimes V_{B,\sigma}\ket{\tilde{\psi}}}^2,\\ \ket{\tilde{\psi}_{\rho,\sigma}} &= \begin{cases}\frac{1}{\sqrt{p_{\rho,\sigma}}}V_{A,\rho}\otimes V_{B,\sigma}\ket{\tilde{\psi}} &\quad p_{\rho,\sigma} > 0,\\
0 &\quad p_{\rho,\sigma} = 0,
\end{cases}
\end{align*}
and consider strategies $$\mathcal{S}_{I\otimes\rho,I\otimes\sigma, \ket{\tilde{\psi}_{\rho,\sigma}}} = \paren{\{I_{d_A d_\rho}\otimes\rho(P_i)\},\{I_{d_B d_\sigma}\otimes\sigma(Q_j)\},\ket{\tilde{\psi}_{\rho,\sigma}}}.$$ Using (\ref{eq:winning probability in bias form}), we can write
\begin{align*}
\nu(\mathcal{G},\tilde{\mathcal{S}}) &= \bra{\tilde{\psi}}\Paren{\sum_{i\in I_A,j \in I_B} \lambda_{ij}\tilde{A}_i \otimes \tilde{B}_j +\sum_{i\in I_A} \lambda_{i}\tilde{A}_i \otimes I+ \sum_{j \in I_B} \lambda_{j}I \otimes \tilde{B}_j + \lambda I\otimes I}\ket{\tilde{\psi}}\\
&= \sum_{\rho,\sigma}\bra{\tilde{\psi}}V_{A,\rho}^\ast \otimes V_{B,\sigma}^\ast\Bigl(\sum_{i\in I_A,j \in I_B} \lambda_{ij} \paren{I_{d_A d_\rho}\otimes \rho({P}_i)}\otimes \paren{ I_{d_B d_\sigma}\otimes\sigma({Q}_j)} +\sum_{i\in I_A} \lambda_{i}\paren{I_{d_A d_\rho}\otimes \rho({P}_i)} \otimes I \\
& \quad \quad \quad \quad \quad \quad \quad \quad \quad \quad+  \sum_{j \in I_B} \lambda_{j}I \otimes  \paren{I_{d_B d_\sigma}\otimes\sigma({Q}_j)} + \lambda I\otimes I\Bigr)V_{A,\rho}\otimes V_{B,\sigma}\ket{\tilde{\psi}}\\
&= \sum_{\rho,\sigma} p_{\rho,\sigma} \nu(\mathcal{G},\mathcal{S}_{I\otimes\rho,I\otimes\sigma, \ket{\tilde{\psi}_{\rho,\sigma}}}).
\end{align*}
Note that $\sum_{\rho,\sigma} p_{\rho,\sigma} = 1$.
In other words, the winning probability of $\tilde{\mathcal{S}}$ is a convex combination of the winning probabilities of irreducible strategies $\mathcal{S}_{I\otimes\rho,I\otimes\sigma, \ket{\tilde{\psi}_{\rho,\sigma}}}$. It is easily verified that $\nu(\mathcal{G},\mathcal{S}_{I\otimes\rho,I\otimes\sigma, \ket{\tilde{\psi}_{\rho,\sigma}}}) \leq \nu(\mathcal{G},\rho,\sigma)$. By assumption of the lemma $\nu(\mathcal{G},\rho,\sigma) < \nu^\ast(\mathcal{G})$ except when $(\rho,\sigma) = (\bar{\rho},\bar{\sigma})$. Now since $\tilde{\mathcal{S}}$ is an optimal strategy, we have $$p_{\rho,\sigma}= \begin{cases} 1 &\quad (\rho,\sigma) = (\bar{\rho},\bar{\sigma}),\\0 &\quad \text{otherwise}.\end{cases}$$
Therefore $\nu(\mathcal{G},\tilde{\mathcal{S}}) = \nu(\mathcal{G},\mathcal{S}_{I\otimes\rho,I\otimes\sigma, \ket{\tilde{\psi}_{\rho,\sigma}}})$ and hence $\mathcal{S}_{I\otimes\bar{\rho},I\otimes\bar{\sigma}, \ket{\tilde{\psi}_{\bar{\rho},
\bar{\sigma}}}}$ is an optimal strategy. From the assumption of the lemma ,$\ket{\psi}$ is the unique state optimizing the strategy induced by $(\bar{\rho},\bar{\sigma})$. Therefore $\ket{\tilde{\psi}_{\bar{\rho},\bar{\sigma}}} = \ket{\text{junk}'}\ket{\psi}$ where both $\ket{\text{junk}'}$ and $\ket{\psi}$ are shared between Alice and Bob such that $\ket{\text{junk}'}$ is the state of the register upon which the identities of Alice and Bob in the operators $(I\otimes \rho)_A \otimes (I\otimes \sigma)_B$ are applied.
In summary
\begin{align}\label{eq:states}
\ket{\tilde{\psi}_{\rho,\sigma}}= \begin{cases} \ket{\text{junk}'}\ket{\psi} &\quad (\rho,\sigma) = (\bar{\rho},\bar{\sigma}),\\0 &\quad \text{otherwise}.\end{cases}
\end{align}

Now using (\ref{eq:first consequence of gh 1}), it follows that
\begin{align*}
\tilde{A}_i\otimes V_B \ket{\tilde{\psi}} &= V_{A}^\ast g_A(P_i) V_{A} \otimes V_B\ket{\tilde{\psi}},
\end{align*} 
from which
\begin{align*}
V_A\tilde{A}_i\otimes V_B \ket{\tilde{\psi}} &= V_AV_{A}^\ast g_A(P_i) V_{A} \otimes V_B\ket{\tilde{\psi}}.
\end{align*} 
Since $V_A V_A^\ast$ is a projection and $V_A\tilde{A}_i\otimes V_B \ket{\tilde{\psi}}$ and $g_A(P_i) V_{A} \otimes V_B\ket{\tilde{\psi}}$ are both unit vectors, it holds that
\begin{align*}
V_A\tilde{A}_i\otimes V_B \ket{\tilde{\psi}} &= g_A(P_i) V_{A} \otimes V_B\ket{\tilde{\psi}}\\
&=  \bigoplus_{\rho,\sigma}\paren{I_{d_A d_\rho}\otimes\rho(P_i)}\otimes I_{d_B d_\sigma^2}\ket{\tilde{\psi}_{\rho,\sigma}}\\
&= \Paren{\ket{\text{junk}'}\bar{\rho}(P_i)\otimes I_{d_{\bar{\sigma}}}\ket{\psi}}\oplus_{(\rho,\sigma)\neq (\bar{\rho},\bar{\sigma})} 0_{d_A d_{\rho}^2 d_B d_\sigma^2}\\
&= \ket{\text{junk}}\bar{\rho}(P_i)\otimes I_{d_{\bar{\sigma}}}\ket{\psi},
\end{align*} 
where the third equality follows from (\ref{eq:states}), and in the fourth equality $\ket{\text{junk}} = \ket{\text{junk}'}\oplus 0$ where $0 \in \complex^{d_Ad_B\paren{\frac{\abs{G_A}\abs{G_B}}{d_{\bar{\rho}} d_{\bar{\sigma}}}- d_{\bar{\rho}} d_{\bar{\sigma}}}}$. Note that $d_Ad_B\paren{\frac{\abs{G_A}\abs{G_B}}{d_{\bar{\rho}} d_{\bar{\sigma}}}- d_{\bar{\rho}} d_{\bar{\sigma}}}$ is a positive integer because the degree of an irreducible representation divides the order of the group.
\end{proof}

\begin{cor}\label{corollary: application of gh}
If in addition to the assumptions of Lemma \ref{lemma:application_of_gh}, it holds that for every optimal strategy $\tilde{\mathcal{S}} = (\{\tilde{A}_i\},\{\tilde{B}_j\},\ket{\tilde{\psi}})$, $f_A^{\tilde{\mathcal{S}}}$ and $f_B^{\tilde{\mathcal{S}}}$ are $\ket{\tilde{\psi}}$-representations, then $\mathcal{G}$ is a self-test for the strategy $\mathcal{S}_{\bar{\rho},\bar{\sigma},\ket{\psi}}$.
\end{cor} 

Note that all these results can be stated robustly using the notion of ($\epsilon,\ket{\psi}$)-representation, but in this paper we focus our attention on exact rigidity. In this paper we use SOS to obtain the extra assumption of Corollary \ref{corollary: application of gh} as seen in Sections \ref{sec:sos_prelim} and \ref{Upperbounds}.

\section{A generalization of CHSH}\label{Zn Games}

The CHSH game can also be viewed as an LCS game where the linear system, over multiplicative $\integer_2$, is given by
\begin{align*}
    &x_0x_1=1,  \\
    &x_0x_1=-1.
\end{align*}
The CHSH viewed as an LCS is first considered in \cite{BCSTensor}. We generalize this to a game $\mathcal{G}_n$ over $\integer_n$ for each $n \geq 2$
\begin{align*}
    &x_0x_1=1,  \\
    &x_0x_1=\omega_n.
\end{align*}
As is the case for $\mathcal{G}_2 = CHSH$, the classical value of $\mathcal{G}_n$ is easily seen to be $0.75$. In Section \ref{LowerBounds}, we exhibit quantum advantage by presenting a strategy $\mathcal{S}_n$ showing that $\nu^*(\mathcal{G}_n) \geq \nu(\mathcal{G}_n,\mathcal{S}_n) = \frac{1}{2} + \frac{1}{2n\sin\Paren{ \frac{\pi}{2n} } } > \frac{1}{2} + \frac{1}{\pi} \approx 0.81$. In Section \ref{Group generated by solutions}, we present the group $G_n$ generated by the operators in $\mathcal{S}_n$.
In Section \ref{Upperbounds}, we show that $\mathcal{G}_3$ is a self-test, and conjecture that this is true for all $n\geq 2$.

As defined in the preliminaries, conventionally, in an LCS game, Alice has to respond with an assignment to all variables in her equation. It is in Alice's best interest to always respond with a satisfying assignment. Therefore, the referee could always determine Alice's assignment to $x_1$ from her assignment to $x_0$. Hence, without loss of generality, in our games, Alice only responds with an assignment to $x_0$. 

Formally $\mathcal{G}_n = ([2],[2],\integer_n,\integer_n,\pi,V)$ where $\integer_n = \{1,\omega_n,\ldots,\omega_n^{n-1}\}$, $\pi$ is the uniform distribution on $[2]\times [2]$, and
\begin{align*}
    V(0,0,a,b) &= 1 \iff a = b,\\
    V(0,1,a,b) &= 1 \iff ab = 1,\\
    V(1,0,a,b) &= 1 \iff a = b,\\
    V(1,1,a,b) &= 1 \iff ab = \omega_n. 
\end{align*}

Consider the quantum strategy $\mathcal{S}$ given by the state $\ket{\psi}$, and projective measurements $\lbrace E_{0,a} \rbrace_{a\in [n]}$ and $\lbrace E_{1,a} \rbrace_{a\in [n]}$ for Alice, and $\lbrace F_{0,b} \rbrace_{b\in [n]}$ and $\lbrace F_{1,b} \rbrace_{b\in [n]}$ for Bob. Note that in our measurement systems, we identify outcome $a\in [n]$ with answer $\omega_n^a \in \integer_n$. As done in the preliminaries, define the generalized observables $A_0= \sum^{n-1}_{i=0} \omega_n^{i}E_{0,i}, A_1 = \sum^{n-1}_{i=0} \omega_n^{i}E_{1,i}, B_0= \sum^{n-1}_{i=0} \omega_n^{i}F_{0,i}, B_1 = \sum^{n-1}_{i=0} \omega_n^{i}F_{1,i}$. We derive an expression for the winning probability of this strategy in terms of the these generalized observables. We do so by introducing the bias operator

\[\mathcal{B}_n = \mathcal{B}_n(A_0,A_1,B_0,B_1) = \sum_{i=1}^{n-1} A_{0}^{i} B_{0}^{-i} + A_{0}^{i}  B_{1}^{i} + A_{1}^{i} B_{0}^{-i}  + \omega_n^{-i} A_{1}^{i}  B_{1}^{i},\]
in which we dropped the tensor product symbol between Alice and Bob's operators.

\begin{prop}\label{Bias equation theorem}
Given the strategy $\mathcal{S}$ above, it holds that $ \nu \left(\mathcal{G}_n, \mathcal{S} \right) = \frac{1}{4n}\bra{\psi}\mathcal{B}_n\ket{\psi} + \frac{1}{n}$.
\end{prop}

\begin{proof}
\begin{align*}
\mathcal{B}_n + 4I &= \sum_{i=0}^{n-1} A_{0}^{i}  B_{0}^{-i} + A_{0}^{i}  B_{1}^{i} + A_{1}^{i}  B_{0}^{-i} + \omega_n^{-i} A_{1}^{i}  B_{1}^{i}\\
&=\sum_{i=0}^{n-1}\sum_{a,b=0}^{n-1}  \omega_n^{i(a-b)}E_{0,a}F_{0,b} + \omega_n^{i(a+b)}E_{0,a}F_{1,b} + \omega_n^{i(a-b)}E_{1,a}F_{0,b} + \omega_n^{i(a + b-1)}E_{1,a}F_{1,b}\\
&=\sum_{a,b=0}^{n-1}\sum_{i=0}^{n-1}  \omega_n^{i(a-b)}E_{0,a}F_{0,b} + \omega_n^{i(a+b)}E_{0,a}F_{1,b} + \omega_n^{i(a-b)}E_{1,a}F_{0,b} + \omega_n^{i(a + b - 1)}E_{1,a}F_{1,b}\\
&=n \sum_{a=0}^{n-1} E_{0,a} F_{0,a}+E_{0,a} F_{1,-a}+E_{1,a} F_{0,a}+E_{1,a} F_{1,1-a}
\end{align*}
in which in the last equality we used the identity $1+\omega_n+\ldots+\omega_n^{n-1}=0$. Also note that in $F_{1,-a}$ and $F_{1,1-a}$ second indices should be read mod $n$. Finally notice that $$\nu(\mathcal{G}, \mathcal{S}) = \frac{1}{4} \bra{\psi}\Paren{\sum_{a=0}^{n-1} E_{0,a} F_{0,a}+E_{0,a} F_{1,-a}+E_{1,a} F_{0,a}+E_{1,a} F_{1,1-a} }\ket{\psi}.$$
\end{proof}

\section{Strategies for \texorpdfstring{$\mathcal{G}_n$}{Gn}}\label{LowerBounds}

In this section, we present quantum strategies $\mathcal{S}_n$ for $\mathcal{G}_{n}$ games. 
In Section \ref{Analysis of winning strategy}, we show that $\nu(\mathcal{G}_n,\mathcal{S}_n) = \frac{1}{2} + \frac{1}{2n\sin\Paren{ \frac{\pi}{2n} } }$ and that this value approaches $\frac{1}{2} + \frac{1}{\pi}$ from above as $n$ tends to infinity. This lower bounds the quantum value $\nu^\ast(\mathcal{G}_n)$, and proves that these games exhibit quantum advantage with a constant gap $> \frac{1}{\pi}-\frac{1}{4}$. We also show that the states in these strategies have full-Schmidt rank. Furthermore the states tend to the maximally entangled state as $n\to\infty$.

We conjecture that $\mathcal{S}_n$ are optimal and that the games $\mathcal{G}_n$ are self-tests for $\mathcal{S}_n$. In Section \ref{Upperbounds}, we prove this for $n = 3$. Using the NPA hierarchy, we verify the optimality numerically up to $n=7$. If the self-testing conjecture is true, we have a family of games with one bit questions and $\log(n)$ bits answers, that self-test entangled states of local dimension $n$ for any $n$. 

\subsection{Definition of the strategy}\label{Observables}
Let $\sigma_{n} = (0 \, 1 \, 2 \, \dots \, n-1) \in S_{n}$ denote the cycle permutation that sends $i$ to $i+1\Mod{n}$. Let $z_{n} = \omega_{n}^{1/4} =  e^{i \pi/2n}$. Let $D_{n,j} = I_n - 2 e_j e_j^\ast$ be the diagonal matrix with $-1$ in the $(j,j)$ entry, and $1$ everywhere else in the diagonal. Then let $D_{n,S} := \prod_{j \in S} D_{n,j}$, where $S \subset [n]$. Finally, let $X_{n}$ be the shift operator (also known as the generalized Pauli $X$), i.e.,  $X_n e_i = e_{\sigma_n(i)}$. For convenience, we shall often drop the $n$ subscript when the dimension is clear from context, and so just refer to $z_n, D_{n,j}, D_{n,S}, X_n$ as $z, D_{j}, D_{S}, X$, respectively. 

Let $\mathcal{H}_A=\mathcal{H}_B = \complex^n$. Then Alice and Bob's shared state in $\mathcal{S}_n$ is defined to be
\[ \ket{\psi_n} = \frac{1}{\gamma_n} \sum_{i=0}^{n-1}(1 - z^{n+2i+1})\ket{\sigma^i(0), \sigma^{-i}(0)} \in \mathcal{H}_A\otimes \mathcal{H}_B,\]
where $\gamma_n = \sqrt{2n + \frac{2}{\sin\Paren{\frac{\pi}{2n}}}}$ is the normalization factor. The generalized observables in $\mathcal{S}_n$ are
\begin{align*}
    A_{0} &= X \\
    A_{1} &= z^{2} D_{0}X \\
    B_{0} &= X \\
    B_{1} &= z^{2} D_{0}X^{*}.
\end{align*}

\begin{ex} \label{example:CHSH}
In $\mathcal{S}_{2}$, Alice and Bob's observables are 
\begin{gather*}
    A_{0} = \sigma_x = \begin{pmatrix} 0 & 1 \\ 1 & 0 \end{pmatrix}, \quad A_{1} = \sigma_y = \begin{pmatrix} 0 & -i \\ i & 0 \end{pmatrix}, \\
    B_{0} = \sigma_x = \begin{pmatrix} 0 & 1 \\ 1 & 0 \end{pmatrix}, \quad B_{1} = \sigma_y = \begin{pmatrix} 0 & -i \\ i & 0 \end{pmatrix},
\end{gather*}
and their entangled state is 
\[ \ket{\psi_2} = \frac{1}{\sqrt{ 4 + 2\sqrt{2} }} \Paren{ \Paren{1 + \frac{1-i}{\sqrt{2}}}\ket{00} - \Paren{1 + \frac{1 + i}{\sqrt{2}}}\ket{11} }. \]

One can verify that this indeed give us the quantum value for CHSH $\frac{1}{2} + \frac{\sqrt{2}}{4}$. 
\end{ex}

\begin{ex}\label{z3_example}

In $\mathcal{S}_{3}$, Alice and Bob's observables are
\begin{gather*}
    A_{0} = \begin{pmatrix} 
0 & 0 & 1 \\
1 & 0 & 0 \\
0 & 1 & 0
\end{pmatrix}, \quad A_{1} = \begin{pmatrix} 
0 & 0 & -z^2 \\
z^2 & 0 & 0 \\
0 & z^2 & 0
\end{pmatrix}, \\
B_{0} = \begin{pmatrix} 
0 & 0 & 1 \\
1 & 0 & 0 \\
0 & 1 & 0
\end{pmatrix}, \quad B_{1} = \begin{pmatrix} 
0 & -z^2 & 0 \\
0 & 0 & z^2 \\
z^2 & 0 & 0
\end{pmatrix},
\end{gather*}
with the entangled state
\[ \ket{\psi_3} = \frac{1}{\sqrt{10}} \Paren{ (1 - z^{4})\ket{00} + 2\ket{12} + (1 + z^{2})\ket{21} }. \] 
One can compute that $\bra{\psi}\mathcal{B}_3\ket{\psi} = 6$. Hence, by Proposition \ref{Bias equation theorem}, we have $\nu^\ast(\mathcal{G}_{3}) \geq \frac{5}{6}$.

\end{ex}

\subsection{Analysis of the strategy}\label{Analysis of winning strategy}

In this section, we prove that $\mathcal{S}_n$ is a quantum strategy and calculate its winning probability. We then prove that the entanglement entropy of $\ket{\psi_n}$ approaches the maximum entropy as $n$ tends to infinity. 

\begin{prop} \label{lowerbound_lem_1}
For $n \in \mathbb{N}$, it holds that $\sum_{j=0}^{n-1} z_n^{2j+n+1} = \sum_{j=0}^{n-1} z_n^{-(2j+n+1)}$.
\end{prop}
\begin{proof}
A direct computation gives
\[ \sum_{j=0}^{n-1} z^{2j+n+1} = \frac{2z^{n+1}}{1-z^{2}} = \frac{2z^{-n-1}}{1-z^{-2}} = \sum_{j=0}^{n-1} z^{-(2j+n+1)}, \]
where we have used the fact that $z^{2n} = -1$.
\end{proof}

\begin{prop}
\label{lowerbound_lem_2}
For $n \in \mathbb{N}$, it holds that $\sum_{j = 0}^{n-1} z_n^{2j+n+1} = -\frac{1}{ \sin\Paren{ \frac{\pi}{2n} } }.$
\end{prop}
\begin{proof}
We handle the even and odd case separately, and in both cases we use the well-known identity for the Dirichlet kernel mentioned in preliminaries. For odd $n$
\begin{align*}
-\sum_{j=0}^{n-1} z^{2j+n+1} &= \sum_{j=0}^{n-1} z^{2j-(n-1)} = \sum_{j=-\frac{n-1}{2}}^{\frac{n-1}{2}} z^{2j} = \sum_{j=-\frac{n-1}{2}}^{\frac{n-1}{2}} e^{\frac{\pi ij}{n}}\\ &= 2\pi \mathcal{D}_{\frac{n-1}{2}}\Paren{\frac{\pi}{n}} = \frac{ \sin\Paren{ \Paren{ \frac{n-1}{2} + \frac{1}{2} } \frac{\pi}{n} } }{\sin\Paren{ \frac{\pi}{2n} } } = \frac{ 1 }{ \sin\Paren{ \frac{\pi}{2n} } }.
\end{align*}
For even $n$
\begin{align*}
    -\sum_{j=0}^{n-1} z^{2j+n+1} &= z \sum_{j=0}^{n} z^{2j - n} - z^{n+1} = z \sum_{j= -\frac{n}{2}}^{\frac{n}{2}} z^{2j} - z^{n+1} = 2\pi z \mathcal{D}_{\frac{n}{2}} \Paren{\frac{\pi}{n}} - z^{n+1} \\
    & = \Paren{ \cos\Paren{ \frac{\pi}{2n} } + i \sin\Paren{ \frac{\pi}{2n} } } \frac{ \sin\Paren{ \Paren{ \frac{n}{2} + \frac{1}{2} }\frac{\pi}{n} }  }{ \sin\Paren{ \frac{\pi}{2n} } } - i\Paren{ \cos\Paren{ \frac{\pi}{2n} } + i \sin\Paren{ \frac{\pi}{2n} } } \\
    &= \frac{ \cos^{2}\Paren{ \frac{\pi}{2n} } + \sin^{2}\Paren{ \frac{\pi}{2n} } }{ \sin\Paren{ \frac{\pi}{2n} } } = \frac{ 1 }{ \sin\Paren{ \frac{\pi}{2n} } }.
\end{align*}

\end{proof}

Now let's observe a commutation relation between $D_{j}$ and $X^{k}$.

\begin{prop}\label{comm_relation}
$X^{i}D_{j} = D_{\sigma^{i}(j)}X^{i}$, for all $i, j \in [n]$. 
\end{prop}

\begin{proof}
It suffices to prove $XD_{j}=D_{\sigma(j)}X$. We show this by verifying $X D_{j}e_k = D_{\sigma(j)}X e_k$ for all $k \in [n]$. 

\begin{align*}
    X D_{j}e_k  = (-1)^{\delta_{j,k}} e_{\sigma(k)} = (-1)^{\delta_{\sigma(j),\sigma(k)}} e_{\sigma(k)} = D_{\sigma(j)}X e_k
\end{align*}
\end{proof}

Now we prove the strategy defined in section \ref{Observables} is a valid quantum strategy.

\begin{prop} 
$A_{0}, A_{1}, B_{0}, B_{1}$ are order-$n$ generalized observables and $\ket{\psi_n}$ is a unit vector. 
\end{prop}

\begin{proof}

Observe that
\[ A_{0}^{n} = B_{0}^{n} = X^{n} = I,\]
also
\[ A_{1}^{n} = (z^{2} D_{0}X)^{n} = z^{2n} D_{\{0, \sigma^{1}(0), \dots,\sigma^{n-1}(0)\}} X^{n}  = (-1)(-I)I = I.\]
Similarly,
\[ B_{1}^{n} = (z^{2} D_{0} X^{*})^{n} = z^{2n} (X^{*})^{n} D_{\{0,\sigma^{1}(0), \dots,\sigma^{n-1}(0)\}} = (-1)I(-I) = I.\]

It is an easy observation that these operators are also unitary. To see that $\ket{\psi_n}$ is a unit vector write
\begin{align*}
\sum_{i=0}^{n-1} \abs{1-z^{n+2i+1}}^2 &= \sum_{i=0}^{n-1} \Paren{1-\cos\Paren{\frac{\pi(n+2i+1)}{2n}}}^2 +  \sin\Paren{\frac{\pi(n+2i+1)}{2n}}^2\\
&= \sum_{i=0}^{n-1} 2\Paren{1-\cos\Paren{\frac{\pi(n+2i+1)}{2n}}}\\
&= 2n - \sum_{i=0}^{n-1} \Re(z^{n+2i+1})\\
&= 2n + \frac{2}{\sin(\pi/2n)}\\
&= \gamma_n^2,
\end{align*}
where we have used Proposition \ref{lowerbound_lem_2} in the third equality.

\end{proof}

\begin{lem}

The entangled state $\ket{\psi}$ is an eigenvector for the bias $\mathcal{B} = \sum_{j=1}^{n-1} A_{0}^{j} B_{0}^{-j} + A_{0}^{j} B_{1}^{j} + A_{1}^{j} B_{0}^{-j} + z^{-4j} A_{1}^{j} B_{1}^{j}$ with eigenvalue $2n - 4 + \frac{2}{\sin\Paren{ \frac{\pi}{2n} } }$.

\end{lem}

\begin{proof}
For the sake of brevity, we drop the normalization factor $\gamma_n$ in the derivation below, and let $\ket{\varphi} = \gamma_n\ket{\psi_n}$. We write

\begin{align*}
\mathcal{B}\ket{\varphi} &= \Paren{\sum_{j=1}^{n-1} A_{0}^{j} \otimes B_{0}^{-j} + A_{0}^{j} \otimes B_{1}^{j} + A_{1}^{j} \otimes B_{0}^{-j} + z^{-4j} A_{1}^{j} \otimes B_{1}^{j} }\ket{\varphi} \\
&= \Paren{\sum_{j=1}^{n-1} \Paren{X \otimes X^{*}}^{j} + z^{2j} \Paren{X \otimes D_{0} X^{*}}^{j} + z^{2j}\paren{D_{0} X \otimes X^{*}}^{j} + \Paren{D_{0} X \otimes D_{0} X^{*}}^{j} }\ket{\varphi}.
\end{align*}

\begin{lem}  $\Paren{X \otimes D_{0} X^{*}}^{j} \ket{\varphi} = \paren{D_{0} X \otimes X^{*}}^{j}\ket{\varphi}$ and $\Paren{X \otimes X^{*}}^{j} \ket{\varphi} = \paren{D_{0} X \otimes D_{0} X^{*}}^{j}\ket{\varphi}$.
\end{lem}
\begin{proof}
It suffices to show these identities for $j=1$ on states $\ket{\sigma^{i}(0),\sigma^{-i}(0)}$, for all $i$, in place of $\ket{\varphi}$. The result then follows by simple induction. In other words, we prove
\begin{align*}
\Paren{X \otimes D_{0} X^{*}} \ket{\sigma^{i}(0),\sigma^{-i}(0)} = \paren{D_{0} X \otimes X^{*}}\ket{\sigma^{i}(0),\sigma^{-i}(0)},\\
\Paren{X \otimes X^{*}} \ket{\sigma^{i}(0),\sigma^{-i}(0)} = \paren{D_{0} X \otimes D_{0}X^{*}}\ket{\sigma^{i}(0),\sigma^{-i}(0)}.
\end{align*}
Note that $I\otimes D_0 \ket{\sigma^{i+1}(0),\sigma^{-i-1}(0)} = D_0 \otimes I \ket{\sigma^{i+1}(0),\sigma^{-i-1}(0)}$ since $-i-1 = 0 \Mod{n}$ iff $i+1 = 0\Mod{n}$. Therefore
\begin{align*}
    \Paren{X \otimes D_{0} X^{*}} \ket{\sigma^{i}(0),\sigma^{-i}(0)} &=
    \Paren{I \otimes D_{0}} \ket{\sigma^{i+1}(0),\sigma^{-i-1}(0)}\\ &=
    \Paren{D_{0} \otimes I} \ket{\sigma^{i+1}(0),\sigma^{-i-1}(0)}\\ &= \paren{D_{0} X \otimes X^{*}}\ket{\sigma^{i}(0),\sigma^{-i}(0)}.
\end{align*}
The other identity follows similarly.
\end{proof}
Now we write
\begin{align*}
\mathcal{B}\ket{\varphi} &= 2\Paren{\sum_{j=1}^{n-1} \Paren{X \otimes X^{*}}^{j} + z^{2j}\paren{D_{0} X \otimes X^{*}}^{j} }\ket{\varphi}\\
&= 2\sum_{j=1}^{n-1} \Paren{1 + z^{2j}\paren{D_{[j]} \otimes I }}(X \otimes X^*)^j\ket{\varphi}\\
&= 2\sum_{j=1}^{n-1}\sum_{i=0}^{n-1} \Paren{1 - z^{2i+n+1} } \Paren{1 + z^{2j}\paren{D_{[j]} \otimes I }}(X \otimes X^*)^j\ket{\sigma^i(0),\sigma^{-i}(0)}\\
&= 2\sum_{j=1}^{n-1}\sum_{i=0}^{n-1} \Paren{1 - z^{2i+n+1} } \Paren{1 + z^{2j}\paren{D_{[j]} \otimes I }}\ket{\sigma^{i+j}(0),\sigma^{-(i+j)}(0)},
\end{align*}
where in the second equality we use Proposition \ref{comm_relation}, and in the third equality we just expanded $\ket{\varphi}$. Note that 
\begin{align*}
\paren{D_{[j]}\otimes I}\ket{\sigma^{i+j}(0),\sigma^{-(i+j)}(0)} = \begin{cases}
-\ket{\sigma^{i+j}(0),\sigma^{-(i+j)}(0)}\quad &i\in [n-j,n-1],\\
\ket{\sigma^{i+j}(0),\sigma^{-(i+j)}(0)}\quad &i\in [0,n-j-1],
\end{cases}
\end{align*}
and we use this to split the sum
\begin{align*}
\mathcal{B}\ket{\varphi} = 2\sum_{j=1}^{n-1}& \Bigg(\sum_{i=0}^{n-j-1}  \Paren{1 - z^{2i+n+1} } \Paren{1 + z^{2j}}\ket{\sigma^{i+j}(0),\sigma^{-(i+j)}(0)}\\&+ \sum_{i=n-j}^{n-1} \Paren{1 - z^{2i+n+1}} \Paren{1 - z^{2j}} \ket{\sigma^{i+j}(0),\sigma^{-(i+j)}(0)}\Bigg)\\
= 2\sum_{i=0}^{n-1}& \Bigg(\sum_{j=1}^{n-i-1}  \Paren{1 - z^{2i+n+1} } \Paren{1 + z^{2j}}\ket{\sigma^{i+j}(0),\sigma^{-(i+j)}(0)}\\&+ \sum_{j=n-i}^{n-1} \Paren{1 - z^{2i+n+1}} \Paren{1 - z^{2j}} \ket{\sigma^{i+j}(0),\sigma^{-(i+j)}(0)}\Bigg),
\end{align*}
and make a change of variable $r = i+j$ to get
\begin{align*}
\mathcal{B}\ket{\varphi} = 2\sum_{i=0}^{n-1}& \Bigg(\sum_{r=i+1}^{n-1}  \Paren{1 - z^{2i+n+1} } \Paren{1 + z^{2(r-i)}}\ket{\sigma^{r}(0),\sigma^{-r}(0)}\\&+ \sum_{r=n}^{n+i-1} \Paren{1 - z^{2i+n+1}} \Paren{1 - z^{2(r-i)}} \ket{\sigma^{r}(0),\sigma^{-r}(0)}\Bigg).
\end{align*}
We have $z^{2(r-i)} = z^{2(r-n+n-i)} = z^{2n}z^{2(r-n-i)}=-z^{2(r-n-i)}$ and $\sigma^r(0) = \sigma^{r+n}(0)$, so by another change of variable in the second sum where we are summing over $r=[n,n+i-1]$ we obtain
\begin{align*}
\mathcal{B}\ket{\varphi} = 2\sum_{i=0}^{n-1}& \Bigg(\sum_{r=i+1}^{n-1}  \Paren{1 - z^{2i+n+1} } \Paren{1 + z^{2(r-i)}}\ket{\sigma^{r}(0),\sigma^{-r}(0)}\\&+ \sum_{r=0}^{i-1} \Paren{1 - z^{2i+n+1}} \Paren{1 + z^{2(r-i)}} \ket{\sigma^{r}(0),\sigma^{-r}(0)}\Bigg) \\
= 2\sum_{i=0}^{n-1} &\Paren{ \sum_{r=0}^{n-1} \Paren{1 - z^{2i+n+1}}\Paren{1 + z^{2(r-i)}}\ket{\sigma^{r}(0) \sigma^{-r}(0)} - 2\Paren{1 - z^{2i+n+1}}\ket{\sigma^{i}(0) \sigma^{-i}(0)} } \\
= 2\sum_{i=0}^{n-1} &\Paren{ \sum_{r=0}^{n-1} \Paren{1 - z^{2i+n+1}}\Paren{1 + z^{2(r-i)}}\ket{\sigma^{r}(0) \sigma^{-r}(0)}} - 4\ket{\varphi}  \\
= 2\sum_{r=0}^{n-1}&\ket{\sigma^{r}(0) \sigma^{-r}(0)} \Paren{ \sum_{i=0}^{n-1} \Paren{1 - z^{2i+n+1}}\Paren{1 + z^{2(r-i)}}} - 4\ket{\varphi}.
\end{align*}
We also have
\begin{align*}
    \sum_{i=0}^{n-1} \Paren{1 - z^{2i+n+1}}\Paren{1 + z^{2(r-i)}} &= \sum_{i=0}^{n-1} 1 - z^{2r+n+1} + z^{2(r-i)} - z^{2i+n+1}\\
    &= \sum_{i=0}^{n-1} 1 - z^{2r+n+1} + z^{2(r-i)} - z^{-(2i+n+1)}\\
    &= (1-z^{2r+n+1})\sum_{i=0}^{n-1} 1 - z^{-(2i+n+1)}\\
    &= \left(n + \frac{1}{\sin(\frac{\pi}{2n})} \right)(1-z^{2r+n+1}),
\end{align*}
where in the second and last equality we used Propositions \ref{lowerbound_lem_1} and \ref{lowerbound_lem_2}, respectively. Putting these together, we obtain
\begin{align*}
\mathcal{B}\ket{\varphi} &= 2\left(n + \frac{1}{\sin(\frac{\pi}{2n})}\right) \sum_{r=0}^{n-1}(1-z^{2r+n+1}) \ket{\sigma^{r}(0) \sigma^{-r}(0)}  - 4\ket{\varphi}\\
&=\left(2n -4 + \frac{2}{\sin(\frac{\pi}{2n})}\right)\ket{\varphi}.
\end{align*}
\end{proof}

\begin{figure}[H]
\centering
\begin{tikzpicture}
\begin{axis}[
    xlabel = $n$,
    ylabel = {$\nu(\mathcal{G}_n, \mathcal{S}_n)$},
    xmin=1.5, xmax=40.5,
    ymin=0.81, ymax=0.860,
    xtick={2,10,20,30,40},
    ytick={.81, 0.82,0.83,0.84,0.85,0.86},
    yticklabel style={/pgf/number format/fixed, /pgf/number format/precision=4}
]
 
\addplot[
    color=black,
    only marks,
    mark=o,
    ]
    coordinates {
   (2,0.853553) (3,0.833333) (4,0.826641) (5,0.823607) (6,0.821975) (7,0.820997) (8,0.820364) (9,0.819932) (10,0.819623) (11,0.819394) (12,0.819221) (13,0.819086) (14,0.818979) (15,0.818892) (16,0.818822) (17,0.818763) (18,0.818714) (19,0.818673) (20,0.818637) (21,0.818607) (22,0.818581) (23,0.818557) (24,0.818537) (25,0.818519) (26,0.818504) (27,0.818490) (28,0.818477) (29,0.818466) (30,0.818455) (31,0.818446) (32,0.818438) (33,0.818430) (34,0.818423) (35,0.818417) (36,0.818411) (37,0.818406) (38,0.818401) (39,0.818396) (40,0.818392)
    };
 
\end{axis}
\end{tikzpicture}
\hskip 5pt
\begin{tikzpicture}
\begin{axis}[
    xlabel = $n$,
    ylabel = {$S_{\psi_n}/\log(n)$},
    xmin=1.5, xmax=40.5,
    ymin=0.99, ymax=1,
    xtick={2,10,20,30,40},
    yticklabel style={/pgf/number format/fixed, /pgf/number format/precision=4},
]
 
\addplot[
    color=black,
    only marks,
    mark=o,
    ]
    coordinates {
    (2,1.000000) (3,0.991159) (4,0.990294) (5,0.990534) (6,0.990951) (7,0.991362) (8,0.991731) (9,0.992053) (10,0.992334) (11,0.992580) (12,0.992796) (13,0.992988) (14,0.993160) (15,0.993315) (16,0.993454) (17,0.993582) (18,0.993698) (19,0.993805) (20,0.993904) (21,0.993995) (22,0.994080) (23,0.994160) (24,0.994234) (25,0.994303) (26,0.994369) (27,0.994431) (28,0.994489) (29,0.994545) (30,0.994597) (31,0.994647) (32,0.994695) (33,0.994740) (34,0.994783) (35,0.994825) (36,0.994864) (37,0.994902) (38,0.994939) (39,0.994974) (40,0.995008)
    };
\end{axis}

\end{tikzpicture}
\caption{The figure on the left illustrates the fast convergence rate of the winning probabilities as they approach the limit  $1/2+1/\pi$. The figure on the right illustrates the ratio of the entanglement entropy to the maximum entanglement entropy of the states for $n \leq 40$.} \label{fig:convergence_rate}
\end{figure}
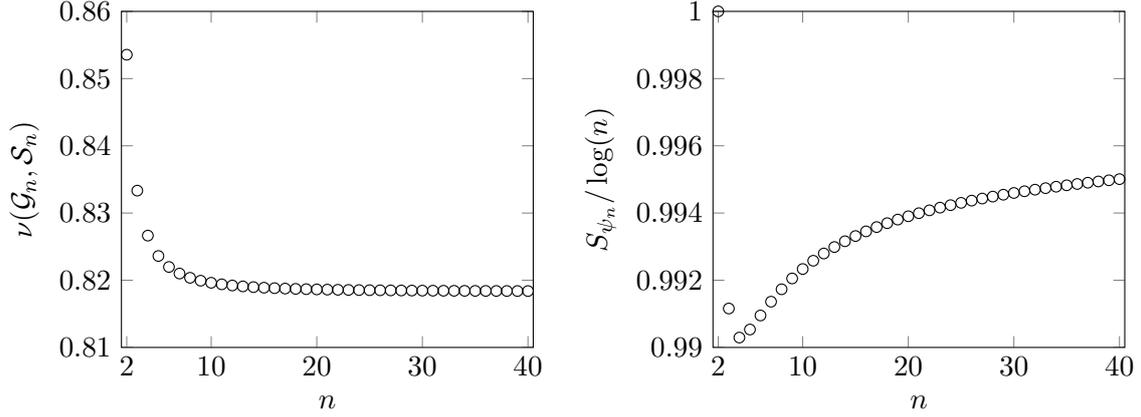

Next we calculate $\nu(\mathcal{G}_n,\mathcal{S}_n)$, its limit as $n$ grows and the entanglement entropy of states $\ket{\psi_n}$. See Figure \ref{fig:convergence_rate}.

\begin{thm}\label{LowerBoundThm}
$\nu(\mathcal{G}_n,\mathcal{S}_n) = \frac{1}{2} + \frac{1}{2n\sin\Paren{ \frac{\pi}{2n} } }.$
\end{thm}
\begin{proof}
\begin{align*}
\nu(\mathcal{G}_n,\mathcal{S}_n) &= \frac{1}{4n}\bra{\psi}\mathcal{B}\ket{\psi} + \frac{1}{n}\\
&= \frac{1}{4n}\bra{\psi} \Paren{ 2n - 4 + \frac{2}{\sin\Paren{\frac{\pi}{2n}} } }\ket{\psi} + \frac{1}{n} \\
&= \frac{1}{4n}\Paren{ 2n - 4 + \frac{2}{\sin\Paren{ \frac{\pi}{2n} } } } + \frac{1}{n} \\
&= \frac{1}{2} + \frac{1}{2n\sin\Paren{ \frac{\pi}{2n} } }. 
\end{align*} 
\end{proof}

\begin{thm} The following hold
\begin{enumerate}
    \item $\lim_{n\to \infty} \nu(\mathcal{G}_n,\mathcal{S}_n) = 1/2+1/\pi$.
    \item $\nu(\mathcal{G}_n,\mathcal{S}_n)$ is a strictly decreasing function.
    \item The games $\mathcal{G}_n$ exhibit quantum advantage, i.e., for $n > 1$ $$\nu^\ast(\mathcal{G}_{n}) > 1/2+1/\pi > 3/4 = \nu(\mathcal{G}_{n}).$$
\end{enumerate}
\end{thm}

\begin{proof}
For the first statement, it suffices to see that 
\begin{align*}
    \lim_{x\to \infty}\frac{1}{2x\sin\Paren{\frac{\pi}{2x} } } = \lim_{x\to \infty}\frac{\frac{1}{2x}}{\sin\Paren{\frac{\pi}{2x} } } = \lim_{x\to \infty}\frac{\frac{-1}{2x^2}}{-\frac{\pi\cos\Paren{\frac{\pi}{2x} }}{2x^2} } = \frac{1}{\pi}.
\end{align*}
For the second statement, we show that the function $f(x) = 2x\sin\paren{\pi/2x}$ is strictly increasing for $x\geq 1$. We have $f'(x) = 2\sin\paren{\pi/2x} - \pi \cos\paren{\pi/2x}/x$. Then $f'(x) > 0$ is equivalent to $\tan\paren{\pi/2x} \geq \pi/2x$. This latter statement is true for all $x \geq 1$. The third statement follows from the first two. 
\end{proof}

\begin{thm}\label{EntropyCorollary}
States $\ket{\psi_n}$ have full Schmidt rank and the ratio of entanglement entropy to maximum entangled entropy, i.e., $S_{\psi_n}/\log(n)$ approaches $1$ as $n \to \infty $. 
\end{thm}
\begin{proof}
Recall that
\[ \ket{\psi_n} = \frac{1}{\gamma_n}\sum_{i = 0}^{n-1} \Paren{ 1 - z^{2i + n + 1}}\ket{ \sigma^{i}(0),\sigma^{-i}(0) } \in \mathcal{H}_{A} \otimes \mathcal{H}_{B}. \]
Let $\ket{i_{A}}= \frac{1-z^{2i + n + 1}}{\Norm{1-z^{2i+n+1}}}\ket{\sigma^{i}(0)}$ and $\ket{i_{B}}= \ket{\sigma^{-i}(0)}$. Clearly $\{i_A\}_i$ and $\{i_B\}_i$ are orthonormal bases for $\mathcal{H}_A$ and $\mathcal{H}_B$, respectively. The Schmidt decomposition is now given by
\[ \ket{\psi_n} = \frac{1}{\gamma_n}\sum_{i=0}^{n-1} \Norm{1-z^{2i+n+1}}\ket{i_{A} i_{B}}. \]
To calculate the limit of $S_{\psi_n}/\log(n)$ first note that
\begin{align*}
     \frac{S_{\psi_n}}{\log(n)} &= -\frac{ \sum_{i=0}^{n-1} \Norm{ 1 - z^{2i+n+1} }^{2}  \log\frac{\Norm{ 1 - z^{2i+n+1} }^{2}}{\gamma_n^2}  }{ \gamma_n^2 \log(n) } \\
    &= -\frac{ \sum_{i=0}^{n-1} \Norm{ 1 - z^{2i+n+1} }^{2} \Paren{ \log\Norm{ 1 - z^{2i+n+1} }^{2} - \log\gamma_n^2 } }{\gamma_n^2 \log(n) } \\
    &\geq  -\frac{ \log(4) \sum_{i=0}^{n-1} \Norm{ 1 - z^{2i+n+1} }^{2} }{  \gamma_n^2\log(n) } + \frac{ \log\gamma_n^2 \sum_{i=0}^{n-1} \Norm{ 1 - z^{2i+n+1} }^{2} }{ \gamma_n^2\log(n) } \\
    &=-\frac{ \log(4)}{  \log(n) } + \frac{ \log\gamma_n^2}{\log(n) }
\end{align*}
where for the inequality we used the fact that $\Norm{ 1 - z^{2i+n+1} } \leq 2$, and for the last equality we used the identity $\gamma_n^2 = \sum_{i=0}^{n-1} \Norm{ 1 - z^{2i+n+1} }^{2}$. So it holds that $$-\frac{ \log(4)}{  \log(n) } + \frac{ \log\gamma_n^2}{\log(n) }\leq \frac{S_{\psi_n}}{\log(n)}\leq 1.$$
By simple calculus $\lim_{n \to\infty}\frac{ \log\gamma_n^2}{\log(n) }-\frac{ \log(4)}{  \log(n) } = 1$. Therefore
by squeeze theorem $\lim_{n \to \infty} \frac{S_{\psi_n}}{\log(n)} = 1$.

\end{proof}
\section{Group structure of \texorpdfstring{$\mathcal{S}_n$}{Sn}}\label{Group generated by solutions}
Let
$H_{n} = \langle A_{0}, A_{1} \rangle$ 
be the group generated by Alice's observables in $\mathcal{S}_n$. Note that since $(A_1 A_0^\ast)^2 = z_n^4 I$, we could equivalently define $H_n = \langle A_{0}, A_{1}, z_n^{4} I \rangle$. Also let
\[ G_{n} = \left\langle P_{0}, P_{1}, J \ | \ P_{0}^{n}, P_{1}^{n}, J^{n}, [J, P_{0}], [J, P_{1}], J^{i}\Paren{ P_{0}^{i} P_{1}^{-i} }^{2} \text{ for } i = 1, 2, \dots, \left\lfloor n/2 \right\rfloor \right\rangle. \]

In this section we show that $H_{n} \cong G_{n}$. So it also holds that $H_{n}$ is a representation of $G_n$. We conjecture that $\mathcal{G}_n$ is a self-test for $G_{n}$, in the sense that every optimal strategy of $\mathcal{G}_n$ is a $\ket{\psi}$-representation of $G_{n}$. In Section \ref{Upperbounds}, we prove this for $n=3$. 

\begin{rmk}
    Note that the relations $J^{i}\Paren{ P_{0}^{i} P_{1}^{-i} }^{2}$ holds in $G_n$ for all $i$. 
    %
\end{rmk}

The following lemma helps us develop a normal form for elements of $G_{n}$.

\begin{lem} \label{comm_lemma} For all $i,j$, the elements $P_{0}^{i}P_{1}^{-i}$ and $P_{0}^{j}P_{1}^{-j}$ commute.
\end{lem}
\begin{proof}
\begin{align*}
    \big(P_{0}^{i}P_{1}^{-i}\big)\big(P_{0}^{j}P_{1}^{-j}\big) &= J^{-i}P_{1}^{i}P_{0}^{-i}P_{0}^{j}P_{1}^{-j} \\
    &= J^{-i}P_{1}^{i}P_{0}^{j-i}P_{1}^{-j} \\
    &= J^{-i}P_{1}^{i} \big( P_{0}^{j-i}P_{1}^{-(j-i)} \big) P_{1}^{-i} \\
    &= J^{-i-(j-i)}P_{1}^{i}P_{1}^{j-i}P_{0}^{-(j-i)}P_{1}^{-i} \\
    &= J^{-j} \big( P_{1}^{j}P_{0}^{-j} \big) \big( P_{0}^{i}P_{1}^{-i} \big) \\
    &= J^{-j}\big( J^{j}P_{0}^{j}P_{1}^{-j} \big) \big( P_{0}^{i}P_{1}^{-i} \big) \\
    &= \big( P_{0}^{j}P_{1}^{-j} \big) \big( P_{0}^{i}P_{1}^{-i} \big).
\end{align*}
\end{proof}

\begin{lem} \label{normform}
For every $g \in G_{n}$ there exist $i, j \in [n]$ and $q_{k} \in \{0, 1\}$ for $k = 1, 2, \dots, n-1$ such that
\[ g = J^{i} P_{0}^{j} \big( P_{0}P_{1}^{-1} \big)^{q_{1}} \big( P_{0}^{2}P_{1}^{-2} \big)^{q_{2}} \cdots \big( P_{0}^{n-1}P_{1}^{-(n-1)} \big)^{q_{n-1}}. \] 
\end{lem}
\begin{proof}
First note that $J$ is central, therefore we can write $g$ in $G_{n}$ as
\[ g = J^{i} P_{0}^{j_{1}} P_{1}^{j_{2}} P_{0}^{j_{3}} \cdots P_{1}^{j_{k}}, \]
for some $k \in \mathbb{N}$, $i \in [n]$, $j_{l} \in [n]$ where $l = 1, 2, \dots, k$. Without loss of generality, let $k$ be even. We perform the following sequence of manipulations
\begin{align*}
    g &= J^{i}P_{0}^{j_{1}}P_{1}^{j_{2}}P_{0}^{j_{3}} \cdots P_{1}^{j_{k-2}}P_{0}^{j_{k-1}} P_{1}^{j_{k}} \\
    &= J^{i}P_{0}^{j_{1}}P_{1}^{j_{2}}P_{0}^{j_{3}} \cdots P_{1}^{j_{k-2}} P_{0}^{j_{k-1}}P_{0}^{j_{k}} \big(P_{0}^{-j_{k}}P_{1}^{j_{k}}\big) \\
    &= J^{i}P_{0}^{j_{1}}P_{1}^{j_{2}}P_{0}^{j_{3}} \cdots P_{1}^{j_{k-2}} P_{1}^{j_{k-1}+j_{k}} \big(P_{1}^{-(j_{k-1}+j_{k})} P_{0}^{j_{k-1}+j_{k}}\big) \big(P_{0}^{-j_{k}}P_{1}^{j_{k}}\big) \\
    &= J^{i-(j_{k-1}+j_{k})}P_{0}^{j_{1}}P_{1}^{j_{2}}P_{0}^{j_{3}} \cdots P_{1}^{j_{k-2}+j_{k-1}+j_{k}} \big(P_{0}^{-(j_{k-1}+j_{k})} P_{1}^{j_{k-1}+j_{k}}\big) \big(P_{0}^{-j_{k}}P_{1}^{j_{k}}\big) \\
    &= \ \cdots \\
    &= J^{i-s}P_{0}^{-s_{1}}\big(P_{0}^{s_{2}}P_{1}^{-s_{2}}\big) \cdots \big(P_{0}^{s_{k-1}}P_{1}^{-s_{k-1}}\big)\big(P_{0}^{s_{k}}P_{1}^{-s_{k}}\big),
\end{align*}
where $s_{l} = -\sum_{t = l}^{k} j_{t}$ and $s = -\sum_{t = 1}^{(k-2)/2} s_{2t+1}$. Then we use the commutation relationship from lemma \ref{comm_lemma} to group the terms with the same $P_{0}$ and $P_{1}$ exponents, and use the relation $J^{i}(P_{0}^{i}P_{1}^{-i})^{2}$ to reduce each term to have an exponent of less than $1$, introducing extra $J$ terms as needed. Finally after reducing the exponents of $J$ and $P_{0}$, knowing that they are all order $n$, we arrive at the desired form. 
\end{proof}

\begin{cor}
$|G_{n}| \leq n^{2} 2^{n-1}$ for all $n \in \mathbb{N}$.
\end{cor}
\begin{proof}
Follows from lemma \ref{normform}.
\end{proof}

\begin{lem}
$|H_{n}| \geq n^{2} 2^{n-1}$ for all $n \in \mathbb{N}$.
\end{lem}
\begin{proof}
We lower bound the order of the group $H_{n}$ by exhibiting $n^{2}2^{n-1}$ distinct elements in the group. We divide the proof into cases depending on the parity of $n$.

First note that $z^{2}D_{i} \in H_{n}$ for all $i \in [n]$ since
\[
    z^{-4i} A_{1}^{i}A_{0}^{-i}A_{1}^{i+1}A_{0}^{-(i+1)} = z^{-4i}z^{2i}D_{[i]}X^{i}X^{-i} z^{2(i+1)} D_{[i+1]}X^{i+1}X^{-(i+1)}
    = z^{2}D_{i},
\]
where in the first equality we use Proposition \ref{comm_relation}. This allows us to generate $z^{2} D_{i_{0}} D_{i_{1}} \cdots D_{i_{k-1}}$ if $k$ is odd via
\[ z^{-4 (k-1)/2} (z^{2}D_{i_{0}})(z^{2}D_{i_{1}}) \cdots (z^{2}D_{i_{k-1}}) = z^{2} D_{i_{0}}D_{i_{1}} \cdots D_{i_{k-1}}, \tag{1}\]
and $D_{i_{0}} D_{i_{1}} \cdots D_{i_{k-1}}$ if $k$ is even by
\[ z^{-4(k/2)} (z^{2}D_{i_{0}})(z^{2}D_{i_{1}}) \cdots (z^{2}D_{i_{k-1}}) = D_{i_{0}}D_{i_{1}} \cdots D_{i_{k-1}}. \tag{2}\]

Let $n$ be odd. From $(2)$ we will be able to generate elements of the form $z^{4i} D_{0}^{q_{0}} D_{1}^{q_{1}} \cdots D_{n-1}^{q_{n-1}} X^{j}$ where there are an even number of nonzero $q_{k}$ for $i,j \in [n]$. It should be clear that the elements with $i \neq i^{\prime} \in \{0, 1, \dots, (n-1)/2\}$ will be distinct. For $i > (n-1)/2$, we simply note that we can factor out a $z^{2n} = -1$ and so we get elements of the form $z^{4i^{\prime}+2} D_{0}^{q_{0}} D_{1}^{q_{1}} \cdots D_{n-1}^{q_{n-1}} X^{j}$, where there are an odd number of nonzero $q_{k}$ for $i^{\prime} \in \{0, 1, \dots, (n-3)/2\}$, $j \in [n]$. Each of these will be distinct from each other as, again, the powers of the $n$th root of unity will be distinct, and distinct from the previous case by the parity of the sign matrices. Therefore we are able to lower-bound $|C_{n}|$ by $n^{2}2^{n-1}$.

If $n$ is even, we will still be able to generate elements of the form $z^{4i} D_{0}^{q_{0}} D_{1}^{q_{1}} \cdots D_{n-1}^{q_{n-1}} X^{j}$ where there are an even number of nonzero $q_{k}$ for $i, j \in [n]$. However, note that for $i > (n-2)/2$, we begin to generate duplicates. So from $(1)$ we can generate elements of the form $z^{4i+2} D_{0}^{q_{0}} D_{1}^{q_{1}} \cdots D_{n-1}^{q_{n-1}} X^{j}$ for $i, j \in [n]$ and an odd number of nonzero $q_{k}$. These will be distinct from the previous elements by the parity of the sign matrices but again will begin to generate duplicates after $i > (n-2)/2$. Therefore we have the lower-bound of $\frac{n}{2}n2^{n-1} + \frac{n}{2}n 2^{n-1} = n^{2}2^{n-1}$ elements.
\end{proof}

\begin{lem}
There exists a surjective homomorphism $f : G_{n} \to H_{n}$.
\end{lem}
\begin{proof}
Let us define $f: \{J,P_0,P_1\}\rightarrow H_n$ by $f(J) = z^4 I, f(P_0)=A_0, f(P_1)=A_1$. We show that $f$ can be extended to a homomorphism from $G_n$ to $H_n$. Consider the formal extension $\tilde{f}$ of $f$ to the free group generated by $\{J,P_0,P_1\}$. We know from the theory of group presentations that $f$ can be extended to a homomorphism if and only if $\tilde{f}(r) = I$ for all relation $r$ in the presentation of $G_n$. 

It is clear that $\tilde{f}$ respects the first five relations of $G_n$. Now we check the last family of relations:
\begin{align*}
    \tilde{f}(J^i (P_0^i P_1^{-i})^2) &= z^{4i}(A_{0}^{i}A_{1}^{-i})^{2} \\
    &= z^{4i}(X^{i}z^{-2i}(D_{0}X)^{-i})^{2} \\
    &= (X^{i}X^{-i}D_{[i]})^{2} \\
    &= D_{[i]}^{2} \\
    &= I.
\end{align*}
The homomorphism $f$ is surjective because $A_0,A_1$ generate the group $H_n$.
\end{proof}

\begin{thm}\label{group_theorem}
$H_{n} \cong G_{n}$ for all $n \in \mathbb{N}$.
\end{thm}
\begin{proof}
Since $f$ is surjective, then $n^2 2^{n-1}\leq \abs{H_n} \leq \abs{G_n} \leq n^2 2^{n-1}$. Thus $\abs{H_n} = \abs{G_n}$, so the homomorphism is also injective. 
\end{proof}

\begin{rmk}\label{remark:Bob-group}
What about the group generated by Bob's operators in $\mathcal{S}_n$? We can define
\[ G'_{n} = \left\langle Q_{0}, Q_{1}, J \ | \ Q_{0}^{n}, Q_{1}^{n}, J^{n}, [J, Q_{0}], [J, Q_{1}], J^{i}\Paren{ Q_{0}^{-i} Q_{1}^{-i} }^{2} \text{ for } i = 1, 2, \dots, \left\lfloor n/2 \right\rfloor \right\rangle. \]
and with a similar argument as in Theorem \ref{group_theorem} show that $\langle B_{0}, B_{1}, z_n^{4} I \rangle \cong G'_n$. It is now easily verified that the mapping $P_0\mapsto Q_0^{-1},P_1\mapsto Q_1,J\mapsto J$ is an isomorphism between $G_n$ and $G'_n$. So Alice and Bob's operator generate the same group, that is $\langle A_{0}, A_{1}, z_n^{4} I \rangle = \langle B_{0}, B_{1}, z_n^{4} I \rangle$. The latter fact could also be verified directly.
\end{rmk}

\section{Sum of squares framework}\label{sec:sos_prelim}
In this paper, the sum of squares (SOS) proofs are used to demonstrate that certain non-commutative polynomials are positive semidefinite. We use this approach to upper bound the quantum value of non-local games and to establish rigidity. This approach has been used previously in the literature, e.g., \cite{NPA, TiltedCHSH}. We illustrate the basics of this framework by going over the proof of optimality and rigidity of CHSH. At the end of this section, we extend this method to deal with the complexities of $\mathcal{G}_n$ and similar games.

By Proposition \ref{Bias equation theorem}, the probability of winning $\mathcal{G}_2$ using a strategy consisting of a state $\ket{\psi}$ and observables $A_0, A_1$ for Alice and $B_0, B_1$ for Bob is given by the expression 
\begin{align*}
    \frac{1}{2} + \frac{1}{8} \bra{\psi}\paren{A_0 B_0 + A_0 B_1 + A_1 B_0 - A_1 B_1} \ket{\psi}.
\end{align*}
To prove $\nu^\ast(\mathcal{G}_2) = \frac{1}{2} + \frac{\sqrt{2}}{4}$, we just need to show that
\begin{align*}
    2\sqrt{2}I - \paren{A_0 B_0 + A_0 B_1 + A_1 B_0 - A_1 B_1} \succeq 0,
\end{align*}
for any observables $A_0,A_1,B_0,B_1$. This immediately follows from the following SOS decomposition
\begin{align}\label{eq:CHSH_sos}
    2\sqrt{2}I - \paren{A_0 B_0 + A_0 B_1 + A_1 B_0 - A_1 B_1} = \frac{\sqrt{2}}{4}\paren{A_0 + A_1 - \sqrt{2} B_0}^2 +  \frac{\sqrt{2}}{4}\paren{A_0 - A_1 - \sqrt{2} B_1}^2.
\end{align}
Next we use this SOS and the Gowers-Hatami theorem to establish that CHSH is a self-test for the strategy $\mathcal{S}_2$ given in Example \ref{example:CHSH}. We learned in Section \ref{Group generated by solutions} that $A_0 = B_0 = \sigma_x$ and $A_1 = B_1 =\sigma_y$ generate
\[ G_2 = \left\langle P_{0}, P_{1}, J \ | \ P_{0}^{2}, P_{1}^{2}, J^{2}, [J, P_{0}], [J, P_{1}], J\Paren{ P_{0} P_{1} }^{2}\right\rangle, \]
which is in fact the dihedral group $D_4$ (also known as the Weyl-Heisenberg group).

The strategy $\mathcal{S}_2$ gives a representation of $D_4$ as seen by the homomorphism $J\mapsto -I, P_0 \mapsto A_0$, and $P_1\mapsto A_1$. Our first step in proving rigidity is to show that a weaker statement holds for any optimal strategy $(\{\tilde{A}_0, \tilde{A}_1\},\{\tilde{B}_0, \tilde{B}_1\},\ket{\tilde{\psi}})$ where $\ket{\tilde{\psi}} \in \mathcal{H}_A \otimes \mathcal{H}_B$ and $\mathcal{H}_A = \complex^{d_A}, \mathcal{H}_B = \complex^{d_B}$. More precisely, we show that any optimal strategy gives rise to a $\ket{\tilde{\psi}}$-representation.
By optimality $$\bra{\tilde{\psi}}\paren{2\sqrt{2}I - \paren{\tilde{A}_0 \tilde{B}_0 + \tilde{A}_0 \tilde{B}_1 + \tilde{A}_1 \tilde{B}_0 - \tilde{A}_1 \tilde{B}_1}}\ket{\tilde{\psi}} = 0.$$ Then by (\ref{eq:CHSH_sos})
\begin{align*}
    \tilde{B}_0\ket{\tilde{\psi}} = \frac{1}{\sqrt{2}}\paren{\tilde{A}_0 + \tilde{A}_1}\ket{\tilde{\psi}},\\
    \tilde{B}_1\ket{\tilde{\psi}}= \frac{1}{\sqrt{2}}\paren{\tilde{A}_0 - \tilde{A}_1}\ket{\tilde{\psi}}.
\end{align*}
These then let us derive the state-dependent anti-commutation relation
\begin{align*}
    \paren{\tilde{B}_0 \tilde{B}_1 + \tilde{B}_1 \tilde{B}_0}\ket{\tilde{\psi}} &= \frac{1}{\sqrt{2}}\paren{\tilde{B}_0 \paren{\tilde{A}_0-\tilde{A}_1} + \tilde{B}_1 \paren{\tilde{A}_0+\tilde{A}_1}}\ket{\tilde{\psi}}\\
    &= \frac{1}{\sqrt{2}}\paren{\paren{\tilde{A}_0-\tilde{A}_1}\tilde{B}_0  + \paren{\tilde{A}_0+\tilde{A}_1}\tilde{B}_1}\ket{\tilde{\psi}}\\
    &= \frac{1}{2}\paren{\paren{\tilde{A}_0-\tilde{A}_1}\paren{\tilde{A}_0+\tilde{A}_1}  + \paren{\tilde{A}_0+\tilde{A}_1}\paren{\tilde{A}_0-\tilde{A}_1}}\ket{\tilde{\psi}}\\
    &= 0,
\end{align*}
where in the second equality we used the fact that Alice and Bob's operators commute. Similarly we have that
\begin{align*}
    \paren{\tilde{A}_0 \tilde{A}_1 + \tilde{A}_1 \tilde{A}_0}\ket{\tilde{\psi}} = 0.
\end{align*}

Define the functions $f_A:D_4 \to \Unitary_{d_A}(\complex), f_B:D_4\to \Unitary_{d_B}$ by 
\begin{align*}
f_A(J^i P_0^j P_1^k) &= (-1)^i \tilde{A}_0^j \tilde{A}_1^k,\\
f_B(J^i P_0^j P_1^k) &= (-1)^i \tilde{B}_0^j \tilde{B}_1^k,
\end{align*}
for all $i,j,k \in [2]$. This is well-defined because every element of $D_4$ can be written uniquely as $J^i P_0^j P_1^k$ (See Section \ref{Group generated by solutions}). Next we show that $f_A$ is a $\ket{\tilde{\psi}}$-representation, and a similar argument holds for $f_B$. We show that for all $i_1,j_1,k_1,i_2,j_2,k_2\in [2]$
\begin{align*}
    f_A\paren{J^{i_1} P_0^{j_1} P_1^{k_1}} f_A\paren{J^{i_2} P_0^{j_2} P_1^{k_2}}\ket{\psi} &= f_A\paren{\paren{J^{i_1} P_0^{j_1} P_1^{k_1}} \paren{J^{i_2} P_0^{j_2} P_1^{k_2}}}\ket{\psi} \\&= f_A\paren{J^{i_1+i_2+k_1j_2} P_0^{j_1+j_2} P_1^{k_1+k_2}}\ket{\psi}.
\end{align*}
We prove this as follows
\begin{align*}
    f_A\paren{J^{i_1} P_0^{j_1} P_1^{k_1}} f_A\paren{J^{i_2} P_0^{j_2} P_1^{k_2}}\ket{\psi} &= \paren{(-1)^{i_1} \tilde{A}_0^{j_1} \tilde{A}_1^{k_1}} \paren{(-1)^{i_2} \tilde{A}_0^{j_2} \tilde{A}_1^{k_2}}\ket{\psi}\\
    &= (-1)^{i_1+i_2+k_2j_2} \tilde{A}_0^{j_1} \tilde{A}_1^{k_1+k_2} \tilde{A}_0^{j_2} \ket{\psi}\\
    &= (-1)^{i_1+i_2+k_1j_2} \tilde{A}_0^{j_1+j_2}\tilde{A}_1^{k_1+k_2} \ket{\psi}\\
    &= f_A\paren{J^{i_1+i_2+k_1j_2} P_0^{j_1+j_2} P_1^{k_1+k_2}}\ket{\psi},
\end{align*}
where in lines 2 and 3, we make essential use of the fact that the exponents are modulo $2$.

The representation theory of $D_4$ is simple. There are four irreducible representations of dimension one: These are given by $P_0\mapsto (-1)^i, P_1\mapsto (-1)^j, J\mapsto 1$ for $i,j\in [2]$. The only irreducible representation of dimension larger than one is given by
\begin{align*}
\rho(P_0) &= \sigma_x, \ \rho(P_1) = \sigma_y, \ \rho(J) =-I.
\end{align*}
Among these, $\rho$ is the only irreducible representation that gives rise to an optimal strategy for CHSH. In addition $\ket{\psi_2}$ is the unique state that maximizes $\nu(\CHSH,\mathcal{S}_{\rho,\rho,\ket{\psi}})$. This follows since $\ket{\psi_2}$ is the unique eigenvector associated with the largest eigenvalue of $\mathcal{B}_2(\sigma_x,\sigma_y,\sigma_x,\sigma_y)$. The rigidity of CHSH follows from Corollary \ref{corollary: application of gh}.

Now we propose a general framework for proving rigidity of $\mathcal{G}_n$ and similar games. This framework extends the methods demonstrated in the CHSH example to deal with more complex games. For concreteness, we focus on $\mathcal{G}_n$. We use Corollary \ref{corollary: application of gh} to prove rigidity. This requires us to ascertain two facts about the game $\mathcal{G}$:
\begin{enumerate}
    \item Every optimal strategy induces $\ket{\psi}$-representations of some groups $G_A$ and $G_B$.
    \item There is a unique pair of irreducible representations $\rho,\sigma$ of $G_A,G_B$, respectively, such that $\nu(\mathcal{G},\rho,\sigma) = \nu^\ast(\mathcal{G})$.
\end{enumerate}
The first step is to obtain algebraic relations (i.e., groups $G_A$ and $G_B$) between the observables of optimal strategies. Suppose we found some SOS decomposition 
\[
\lambda_n I - \mathcal{B}_n(a_0,a_1,b_0,b_1) = \sum_{k} T_k(a_0,a_1,b_0,b_1)^\ast T_k(a_0,a_1,b_0,b_1),
\]
where $\mathcal{B}_n$ is the bias polynomial for $\mathcal{G}_n$ and $\lambda_n = 4n\nu^\ast(\mathcal{G}_n) - 4$. This equality is over $$\complex^\ast\langle a_0,a_1,b_0,b_1\rangle/\langle a_i^n-1, b_j^n -1, a_ib_j-a_jb_i : \forall i,j\in \{0,1\}\rangle$$ where $\complex^\ast\langle a_0,a_1,b_0,b_1\rangle$ is the ring of noncommutative polynomials equipped with adjoint, and $\langle a_i^n-1, b_j^n -1, a_ib_j-a_jb_i : \forall i,j\in \{0,1\}\rangle$ is the ideal that forces Alice and Bob's operators to form a valid strategy. 

For any optimal strategy $(\{A_0,A_1\},\{B_0,B_1\},\ket{\psi})$, it holds that 
\[
\bigparen{\lambda_n I-\mathcal{B}_n(A_0,A_1,B_0,B_1)}\ket{\psi} = 0.
\]
So it must also hold that $T_k(A_0,A_1,B_0,B_1)\ket{\psi} = 0$. Let $\paren{M_j(A_0,A_1) -I}\ket{\psi} = 0$ be all the relations derived from $T_k$ such that $M_i$ are monomials only in Alice's operators. Similarly let $\paren{N_j(A_0,A_1) -I}\ket{\psi} = 0$ be all the monomial relations involving only Bob's operators. We call $M_i,N_j$ the \emph{group relations}. Define groups 
\begin{align*}
G_A = \langle P_0, P_1 : M_i(P_0,P_1)\rangle,\quad G_B = \langle Q_0,Q_1: N_j(Q_0,Q_1)\rangle.
\end{align*}
In the case of $\mathcal{G}_n$, we in fact have $G_A=G_B=G_n$.\footnote{In Section \ref{Group generated by solutions}, we gave a presentation for $G_n$ using three generators, but in fact one could obtain a presentation using only two generators.} Next, prove that, for all optimal strategies, the functions $f_A,f_B$ defined by $f_A(P_i)=A_i$ and $f_B(Q_j)=B_j$ (as in the preliminaries) are $\ket{\psi}$-representations of $G_A, G_B$, respectively. 

To prove the second assumption, one approach is the brute force enumeration of irreducible representation pairs. A more practical approach, when dealing with families of games, is to demonstrate uniqueness of the pair of optimal irreducible representations using \emph{ring relations}. Let $R_i(A_0,A_1)\ket{\psi} = 0$ be all the relations derived from $T_k$. We call $R_i(A_0,A_1)\ket{\psi}=0$ ring relations. They are allowed to be arbitrary polynomials (as opposed to monomials in the case of group relations). Similarly let $S_j(B_0,B_1)\ket{\psi} = 0$ be all the relations derived from $T_k$ involving only Bob's operators. Then show that there is a unique irreducible representation $\rho$ of $G_A$ (resp. $\sigma$ of $G_B$) satisfying the ring relations, i.e., $R_i(\rho(P_0),\rho(P_1))=0$ (resp. $S_i(\sigma(Q_0),\sigma(Q_1))=0$). Note that here we require the stronger constraint $R_i(\rho(P_0),\rho(P_1))=0$ as opposed to $R_i(\rho(P_0),\rho(P_1))\ket{\psi}=0$.\footnote{The intuition behind this step is the one-to-one correspondence between the group representations of $G_A$ and the ring representations of the group ring $\complex[G_A]$. The optimal pair of irreducible representations are in fact irreducible representations of rings $\complex[G_A]/\langle R_i(P_0,P_1)\rangle$ and $\complex[G_B]/\langle S_j(Q_0,Q_1)\rangle$.}

In some special cases, e.g., games $\mathcal{G}_n$, there is one ring relation that rules them all. For $\mathcal{G}_n$ there is a unique irreducible representation of $G_n$ satisfying the ring relation $\paren{H_n + (n-2)I}\ket{\psi} = 0$ where \begin{align}
    H_n = H_n(A_0,A_1) = \omega \sum_{i=0}^{n-1} A_0^i A_1 A_0^{(n-i-1)}. \label{eq:ring_eq}
\end{align}

For example in the case of $G_5$, there are $25$ degree one irreducible representations given by $P_0\mapsto \omega_5^i,P_1\mapsto\omega_5^j, J\mapsto \omega^{2(j-i)}$ for all $i,j\in [5]$. There are also $15$ irreducible representations of degree five: For each $i\in [5]$, there are three irreducible representations sending $J\to \omega_5^i I_5$. Among these $40$ irreducible representations only one satisfies the ring relation $\paren{H_5 + 3I}\ket{\psi} = 0$. This unique irreducible representation is one of the three irreducible representations mapping $J\mapsto \omega_5 I_5$.\footnote{Interestingly, cousin games of $\mathcal{G}_5$, defined using systems of equation $x_0 x_1 = 1, x_0,x_1=\omega^i$ for $i\in [5]$, generate the same group $G_5$. For every $i$, the unique optimal irreducible representation strategy is one of the three irreducible representations mapping $J\mapsto \omega_5^i I_5$.}   

In section \ref{BCSwithSOS}, we show that in the special case of pseudo-telepathic games, this framework reduces to the solution group formalism of Cleve, Liu, and Slofstra \cite{BCSCommuting}. The group derived from the SOS is the solution group, and the analogue of the ring relation that hones in on the optimal irreducible representation $\rho$ is the requirement that $\rho(J)\neq I$.

In the next section, we use the SOS framework to give a full proof of the rigidity of $\mathcal{G}_3$. While omitted, the cases of $\mathcal{G}_4, \mathcal{G}_5$ follow similarly. The SOS decompositions of $\mathcal{B}_4,\mathcal{B}_5$ are comparatively long and tedious. 

\section{Optimality and rigidity for \texorpdfstring{$\mathcal{G}_3$}{G3}} \label{Upperbounds}

In this section, we show that $\mathcal{S}_3$ is optimal, and therefore $\nu^\ast(\mathcal{G}_3) = 5/6$. We also show that $\mathcal{G}_3$ is a self-test for the strategy $\mathcal{S}_3$. We obtain these results by obtaining algebraic relations between operators in any optimal strategy using an SOS decomposition for $\mathcal{B}_3$. 

\subsection{Optimality of \texorpdfstring{$\mathcal{S}_3$}{S3}}\label{sec:opt}
For every operator $A_i,B_j$ for which $A_i^3=B_j^3=I$ and $[A_i,B_j]=0$, we have the following SOS decomposition:
\begin{align}
&6I - A_0 B_0^\ast-A_0^\ast B_0-A_0 B_1-A_0^\ast B_1^\ast - A_1 B_0^\ast-A_1^\ast B_0-\omega^\ast A_1 B_1 - \omega A_1^\ast B_1^\ast\nonumber\\
& \ = \lambda_1 (S_1^\ast S_1 + S_2^\ast S_2)+\lambda_2(T_1^\ast T_1 + T_2^\ast T_2)+\lambda_3(T_3^\ast T_3 + T_4^\ast T_4)+\lambda_4(T_5^\ast T_5 + T_6^\ast T_6),\label{SOS_Z3}
\end{align}
where
\begin{align*}
S_1 &= A_0 + \omega A_1 + \omega^\ast B_0 + \omega B_1^\ast,\\
S_2 &= A_0^\ast + \omega^\ast A_1^\ast + \omega B_0^\ast + \omega^\ast B_1,\\
T_1 &= A_0 B_0^\ast+ ai A_0^\ast B_0 - a A_0 B_1 + i A_0^\ast B_1^\ast + a A_1 B_0^\ast- i A_1^\ast B_0 - \omega^\ast A_1 B_1 - ai \omega A_1^\ast B_1^\ast,\\
T_2 &= A_0 B_0^\ast+ ai A_0^\ast B_0 + a A_0 B_1 - i A_0^\ast B_1^\ast - a A_1 B_0^\ast+ i A_1^\ast B_0 - \omega^\ast A_1 B_1 - ai \omega A_1^\ast B_1^\ast,\\
T_3 &= A_0 B_0^\ast- ai A_0^\ast B_0 - a A_0 B_1 - i A_0^\ast B_1^\ast + a A_1 B_0^\ast+ i A_1^\ast B_0 - \omega^\ast A_1 B_1 + ai \omega A_1^\ast B_1^\ast,\\
T_4 &= A_0 B_0^\ast- ai A_0^\ast B_0 + a A_0 B_1 + i A_0^\ast B_1^\ast - a A_1 B_0^\ast- i A_1^\ast B_0 - \omega^\ast A_1 B_1 + ai \omega A_1^\ast B_1^\ast,\\
T_5 &= A_0 B_0^\ast+ b A_0^\ast B_0 - b A_0 B_1 -  A_0^\ast B_1^\ast -b A_1 B_0^\ast-  A_1^\ast B_0 + \omega^\ast A_1 B_1 +b \omega A_1^\ast B_1^\ast,\\
T_6 &= 6I - A_0 B_0^\ast-A_0^\ast B_0-A_0 B_1-A_0^\ast B_1^\ast - A_1 B_0^\ast-A_1^\ast B_0-\omega^\ast A_1 B_1 - \omega A_1^\ast B_1^\ast,
\end{align*}
and
\begin{align*}
\lambda_1 = \frac{5}{86}, \ \lambda_2=\frac{14+\sqrt{21}}{4\cdot 86}, \ \lambda_3 =\frac{14-\sqrt{21}}{4\cdot 86},\ \lambda_4=\frac{7}{86},\\
a = \frac{2\omega + 3\omega^\ast}{\sqrt{7}}, \ b=\frac{3\omega+8\omega^\ast}{7}, \omega = \omega_3.
\end{align*}

This SOS decomposition tells us that $\mathcal{B}_3 \preceq 6I$ in positive semidefinite order. So from Theorem \ref{Bias equation theorem}, it holds that $\nu^\ast(\mathcal{G}_3) \leq 5/6$. Combined with Theorem \ref{LowerBoundThm}, we have $\nu^\ast(\mathcal{G}_3) = 5/6$.

This SOS is obtained from the dual semidefinite program associated with the second level of the NPA hierarchy. Surprisingly, the first level of NPA is not enough to obtain this upper bound, as was the case with CHSH.

\subsection{Algebraic relations}\label{sec:algebraic_relations_G3}
As in Section \ref{sec:sos_prelim}, we derive group and ring relations for optimal strategies of $\mathcal{G}_3$ from the SOS (\ref{SOS_Z3}).
For the rest of this section, let $(A_0,A_1,B_0,B_1,\ket{\psi})$ be an optimal strategy. Then $\bra{\psi}( 6I- \mathcal{B}_3 )\ket{\psi} = 0$. So it also holds that $S_i \ket{\psi} = 0$ and $T_j \ket{\psi} = 0$ for all $i \in [2]$ and $j \in [6]$. Therefore
\begin{align*}
(T_1+T_2+T_3+T_4)\ket{\psi} = 0,\quad (T_1+T_2-T_3-T_4)\ket{\psi} = 0,\\
(T_1-T_2+T_3-T_4)\ket{\psi} = 0,\quad (T_1-T_2-T_3+T_4)\ket{\psi} = 0.
\end{align*}
From which by simplification we obtain the four relations
\begin{alignat}{2}
A_0 B_0^\ast \ket{\psi} &= \omega^\ast A_1 B_1 \ket{\psi}, \quad A_0^\ast B_0 \ket{\psi} &&= \omega A_1^\ast B_1^\ast \ket{\psi},\nonumber\\
 A_0 B_1 \ket{\psi} &= A_1 B_0^\ast \ket{\psi}, \quad \quad A_0^\ast B_1^\ast \ket{\psi} &&= A_1^\ast B_0 \ket{\psi}. \label{pairing}
\end{alignat}
Now from these four relations and the fact that $A_i,B_j$ are generalized observables satisfying $[A_i,B_j]=0$, we obtain
\begin{align}
\omega^\ast A_0^\ast A_1 \ket{\psi} &= B_1^\ast B_0^\ast \ket{\psi}\label{A_0A_1_1}\\
\omega A_0 A_1^\ast \ket{\psi} &= B_1 B_0 \ket{\psi}\label{A_0A_1_2}\\
A_0^\ast A_1 \ket{\psi} &= B_0 B_1 \ket{\psi}\label{A_0A_1_3}\\
A_0 A_1^\ast \ket{\psi} &= B_0^\ast B_1^\ast \ket{\psi}\label{A_0A_1_4}\\
A_1^\ast A_0 \ket{\psi} &= \omega^\ast B_0 B_1 \ket{\psi} \label{A_1A_0_1}\\
A_1 A_0^\ast \ket{\psi} &= \omega B_0^\ast B_1^\ast \ket{\psi} \label{A_1A_0_2}\\
A_1^\ast A_0 \ket{\psi} &= B_1^\ast B_0^\ast \ket{\psi}\label{A_1A_0_3}\\
A_1 A_0^\ast \ket{\psi} &= B_1 B_0 \ket{\psi}.\label{A_1A_0_4}
\end{align}
From the pair of relations (\ref{A_0A_1_1}) and (\ref{A_1A_0_3}) as well as the pair of relations (\ref{A_0A_1_2}) and (\ref{A_1A_0_4}), we obtain the following relations between Alice's observables acting on the state $\ket{\psi}$:
\begin{align}
A_0^\ast A_1 \ket{\psi} &= \omega A_1^\ast A_0 \ket{\psi}, \label{group_rel_1}\\
A_1 A_0^\ast \ket{\psi} &= \omega A_0 A_1^\ast \ket{\psi}. \label{group_rel_2}
\end{align}

Next we prove two propositions regarding $H = H_3 = \omega A_0 A_1 A_0 + \omega A_0^\ast A_1 + \omega A_1 A_0^\ast$ defined in (\ref{eq:ring_eq}).
\begin{prop} \label{prop:sum_of_rings_eqs}
$\paren{H+H^\ast}\ket{\psi} = -2\ket{\psi}$
\end{prop}
\begin{proof} 
We start by writing
\begin{align*}
(\omega B_0^\ast + \omega^\ast B_1 + B_0 B_1^\ast + B_1^\ast B_0) \ket{\psi} &= (\omega^\ast B_0 + \omega B_1^\ast)(\omega^\ast B_0 + \omega B_1^\ast) \ket{\psi}\\
&= -(\omega^\ast B_0 + \omega B_1^\ast)(A_0+\omega A_1) \ket{\psi}\\
&= -(A_0+\omega A_1)(\omega^\ast B_0 + \omega B_1^\ast)\ket{\psi}\\
&= (A_0+\omega A_1)(A_0+\omega A_1) \ket{\psi}\\
&= (A_0^\ast + \omega^\ast A_1^\ast +\omega A_0 A_1 + \omega A_1 A_0)\ket{\psi},
\end{align*}
where for the second and fourth equality, we used the relation $S_1\ket{\psi} = 0$, and for the third equality we used the fact that Alice and Bob's operators commute. Now using $S_2\ket{\psi} = 0$, we obtain
\begin{align}
( B_0 B_1^\ast + B_1^\ast B_0) \ket{\psi} = (2 A_0^\ast + 2 \omega ^\ast A_1^\ast +\omega A_0 A_1 + \omega A_1 A_0)\ket{\psi}. \label{intermediate_1}
\end{align}
Similarly we have
\begin{align}
( B_1 B_0^\ast + B_0^\ast B_1) \ket{\psi} = (2 A_0 + 2 \omega A_1 +\omega^\ast A_0^\ast A_1^\ast + \omega^\ast A_1^\ast A_0^\ast)\ket{\psi}. \label{intermediate_2}
\end{align}

We proceed by simplifying $T_6 \ket{\psi} = 0$ using relations (\ref{pairing}) to obtain
\begin{align*}
(3I - A_0 B_0^\ast - A_0^\ast B_0 - A_0 B_1 - A_0^\ast B_1^\ast)\ket{\psi} = 0.
\end{align*}
Let $P = A_0 B_0^\ast +A_0^\ast B_0 + A_0 B_1 +A_0^\ast B_1^\ast$, and write
\begin{align}
0 &= \bigparen{3I - A_0 B_0^\ast - A_0^\ast B_0 - A_0 B_1 - A_0^\ast B_1^\ast}^\ast \bigparen{3I - A_0 B_0^\ast - A_0^\ast B_0 - A_0 B_1 - A_0^\ast B_1^\ast}\ket{\psi}\nonumber\\
&= \bigparen{13I - 5P + A_0^\ast (B_1 B_0^\ast+B_0^\ast B_1) + A_0 (B_0 B_1^\ast + B_1^\ast B_0) + B_0^\ast B_1^\ast + B_0 B_1 + B_1 B_0 + B_1^\ast B_0^\ast} \ket{\psi}\nonumber\\
&= (-2I + A_0^\ast (B_1 B_0^\ast+B_0^\ast B_1) + A_0 (B_0 B_1^\ast + B_1^\ast B_0) + B_0^\ast B_1^\ast + B_0 B_1 + B_1 B_0 + B_1^\ast B_0^\ast) \ket{\psi},\label{eq:intermediate}
\end{align}
where in the last line, we used $(3I-P)\ket{\psi} = 0$. Using identities (\ref{intermediate_1}) and (\ref{intermediate_2})
\begin{align*}
&\bigparen{A_0^\ast (B_1 B_0^\ast + B_0^\ast B_1) + A_0 (B_0 B_1^\ast + B_1^\ast B_0)}\ket{\psi} \\&\quad= \bigparen{4I+\omega A_0 A_1 A_0 + \omega^\ast A_0^\ast A_1^\ast A_0^\ast+ 2\omega A_0^\ast A_1 +\omega^\ast A_0 A_1^\ast+2\omega^\ast A_0 A_1^\ast+\omega A_0^\ast A_1}\ket{\psi}.
\end{align*}
Transferring Bob's operators to Alice using identities (\ref{A_0A_1_1}-\ref{A_0A_1_4})
\begin{align*}
\bigparen{B_0^\ast B_1^\ast + B_0 B_1 + B_1 B_0 + B_1^\ast B_0^\ast} \ket{\psi} = 
\bigparen{A_0 A_1^\ast + A_0^\ast A_1+\omega A_0 A_1^\ast + \omega^\ast A_0^\ast A_1}\ket{\psi}.
\end{align*}
Plugging these back in (\ref{eq:intermediate})
\begin{align*}
0 &= (2I + \omega A_0 A_1 A_0 + \omega^\ast A_0^\ast A_1^\ast A_0^\ast +(3\omega +\omega^\ast+1)A_0^\ast A_1 + (3\omega^\ast+\omega+1)A_0 A_1^\ast)\ket{\psi}\\
&=(2I + \omega A_0 A_1 A_0 + \omega^\ast A_0^\ast A_1^\ast A_0^\ast +2\omega A_0^\ast A_1 + 2\omega^\ast A_0 A_1^\ast)\ket{\psi}\\
&=(2I + \omega A_0 A_1 A_0 + \omega^\ast A_0^\ast A_1^\ast A_0^\ast +\omega A_0^\ast A_1 + \omega^\ast A_1^\ast A_0 + \omega^\ast A_0 A_1^\ast + \omega A_1 A_0^\ast)\ket{\psi}. \\
&= (2I + H+H^\ast)\ket{\psi},
\end{align*}
where in the first line we used $1+\omega+\omega^\ast = 0$, and in the second line we used identities (\ref{group_rel_1}) and (\ref{group_rel_2}).
\end{proof}

\begin{prop} \label{prop:ring_eq}
$\paren{H + I}\ket{\psi}=\paren{H^\ast + I}\ket{\psi}=0$.
\end{prop}
\begin{proof}
First note
\begin{align}
\bra{\psi}H^\ast H\ket{\psi} &= \bra{\psi} (3I + A_0^\ast A_1^\ast A_0A_1
+A_1^\ast A_0^\ast A_1A_0+A_1^\ast A_0A_1 A_0^\ast+A_0A_1^\ast A_0^\ast A_1\nonumber\\
&\quad\quad\quad\quad\quad\quad\quad\quad\quad\quad\quad+A_0^\ast A_1^\ast A_0^\ast A_1 A_0^\ast+A_0 A_1^\ast A_0 A_1 A_0)\ket{\psi}.\label{intermediate_3} 
\end{align}
Using (\ref{group_rel_1}) and (\ref{group_rel_2}), we have
\begin{align*}
&\bra{\psi} A_0 A_1^\ast A_0^\ast A_1 \ket{\psi} = \omega\bra{\psi}A_0 A_1^\ast A_1^\ast A_0\ket{\psi} = \omega \bra{\psi}A_0 A_1 A_0\ket{\psi},\\
&\bra{\psi} A_0^\ast A_1^\ast A_0^\ast A_1 A_0^\ast \ket{\psi} = \omega \bra{\psi} A_0^\ast A_1^\ast A_0^\ast A_0 A_1^\ast\ket{\psi} = \omega \bra{\psi} A_0^\ast A_1\ket{\psi},
\end{align*}
and using (\ref{A_0A_1_3}) and (\ref{A_1A_0_1})
\begin{align*}
\bra{\psi} A_0^\ast A_1^\ast A_0 A_1\ket{\psi} = \bra{\psi} A_0^\ast A_1 A_1 A_0^\ast A_0^\ast A_1\ket{\psi} = \omega \bra{\psi}B_1^\ast B_0^\ast A_1 A_0^\ast B_0 B_1 \ket{\psi} = \omega \bra{\psi} A_1 A_0^\ast \ket{\psi},
\end{align*}
and taking conjugate transpose of these three we obtain
\begin{align*}
&\bra{\psi} A_1^\ast A_0 A_1 A_0^\ast \ket{\psi} = \omega^\ast \bra{\psi}A_0^\ast A_1^\ast A_0^\ast\ket{\psi},\\
&\bra{\psi} A_0 A_1^\ast A_0 A_1 A_0 \ket{\psi} = \omega^\ast \bra{\psi} A_1^\ast A_0\ket{\psi},\\
&\bra{\psi} A_1^\ast A_0^\ast A_1 A_0\ket{\psi} = \omega^\ast \bra{\psi} A_0 A_1^\ast \ket{\psi}.
\end{align*}
Plugging these back in (\ref{intermediate_3}), we obtain
\begin{align*}
\norm{H \ket{\psi}}^2 &= \bra{\psi}H^\ast H\ket{\psi} \\
&= \bra{\psi} (3I + \omega A_0 A_1 A_0 +\omega A_0^\ast A_1 +\omega A_1 A_0^\ast+\omega^\ast A_0^\ast A_1^\ast A_0^\ast+\omega^\ast A_1^\ast A_0+\omega^\ast A_0 A_1^\ast)\ket{\psi}\\
&=\bra{\psi} (3I + H + H^\ast)\ket{\psi}\\
&=\bra{\psi}I\ket{\psi}\\
&= 1,
\end{align*}
where in fourth equality we used Proposition \ref{prop:sum_of_rings_eqs}. Similarly $\norm{H^\ast \ket{\psi}} = 1$. From $(H+H^\ast)\ket{\psi} = -2\ket{\psi}$ and the fact that $H\ket{\psi}$ and $H^\ast \ket{\psi}$ are unit vectors, we get that $H\ket{\psi} = H^\ast \ket{\psi} = -\ket{\psi}$. 
\end{proof}

\begin{prop}\label{eq:important_1}
$A_0 A_1 A_0\ket{\psi} = \omega A_0^\ast A_1^\ast A_0^\ast \ket{\psi}.$
\end{prop}
\begin{proof}
By Proposition \ref{prop:ring_eq}, $H\ket{\psi} = H^\ast\ket{\psi}$, and by identities (\ref{group_rel_1}), (\ref{group_rel_2}), $(\omega A_0^\ast A_1 + \omega A_1 A_0^\ast)\ket{\psi} = (\omega^\ast A_1^\ast A_0 + \omega^\ast A_0 A_1^\ast)\ket{\psi}$. Putting these together, we obtain $A_0 A_1 A_0\ket{\psi} = \omega A_0^\ast A_1^\ast A_0^\ast \ket{\psi}.$
\end{proof}

\begin{prop}\label{prop:obs_com}
$A_0 A_1^\ast A_0^\ast A_1 \ket{\psi} = A_0^\ast A_1 A_0 A_1^\ast \ket{\psi}$ in other words $A_0 A_1^\ast$ and $A_0^\ast A_1$ commute on $\ket{\psi}$
\end{prop}
\begin{proof}
To see this write
\begin{align*}
A_0 A_1^\ast A_0^\ast A_1 \ket{\psi} &= \omega A_0 A_1^\ast  A_1^\ast A_0\ket{\psi}\\
&= \omega A_0 A_1 A_0\ket{\psi}\\
&= \omega^\ast A_0^\ast A_1^\ast A_0^\ast \ket{\psi}\\
&= \omega^\ast A_0^\ast A_1 A_1 A_0^\ast \ket{\psi}\\
&= A_0^\ast A_1 A_0 A_1^\ast \ket{\psi},
\end{align*}
where in the first line we used \ref{group_rel_1}, in the third line we used \ref{eq:important_1}, and in the fifth line we used \ref{group_rel_2}.

\end{proof}

\subsection{Rigidity of \texorpdfstring{$\mathcal{G}_3$}{G3}}\label{sec:rig}
Suppose $(\{A_0,A_1\},\{B_0,B_1\},\ket{\psi})$ is an optimal strategy for $\mathcal{G}_3$. By Theorem \ref{group_theorem}, we know that the optimal operators of Alice defined in section \ref{Observables} generate the group
\begin{align*}
G_3 = \left\langle J,P_0,P_1: J^3, P_0^3, P_1^3, [J,P_0],[J,P_1],J(P_0 P_1^{-1})^2\right\rangle,
\end{align*}
The same group is generated by Bob's operators as in Remark \ref{remark:Bob-group}. We apply Corollary \ref{corollary: application of gh} with $G_A=G_B=G_3$. In order to do this, we first prove the following lemma stating that every optimal strategy is a $\ket{\psi}$-representation of $G$.
\begin{lem}\label{representability_lemma} 
Let $(\{A_0,A_1\},\{B_0,B_1\},\ket{\psi})$ be an optimal strategy for $\mathcal{G}_3$. Define maps $f_A, f_B: G_3\rightarrow \Unitary_d(\complex)$ by 
\begin{alignat*}{4}
f_A(J)&=\omega_3 I, \ f_A(P_0) &&= A_0, \ f_A(P_0 P_1^{-1}) &&=A_0 A_1^\ast, \ f_A(P_0^{-1} P_1) &&=A_0^\ast A_1\\
f_B(J)&=\omega_3 I, \ f_B(P_0) &&= B_0^\ast, \ f_B(P_0 P_1^{-1}) &&=B_0^\ast B_1^\ast, \ f_B(P_0^{-1} P_1) &&=B_0 B_1
\end{alignat*}
and extend it to all of $G_3$ using the normal form from Lemma \ref{normform}. Then $f_A,f_B$ are $\ket{\psi}$-representations of $G_3$.
\end{lem}

\begin{proof}
These maps are well defined since every element of $G_3$ can be written uniquly as $$J^{i} P_{0}^{j} \big( P_{0}P_{1}^{-1} \big)^{q_{1}} \big( P_{0}^{-1}P_{1} \big)^{q_{2}}$$
for $i,j\in[3],q_1,q_2\in [2]$. All we need is that $f_A(g)f_A(g')\ket{\psi} = f_A(gg')\ket{\psi}$ for all $g,g'\in G_3$. The proof is reminiscent of the proof that $gg'$ can be written in normal form for every $g,g'\in G_3$. Except that we need to be more careful here, since we are dealing with Alice's operators $A_0,A_1$, and not the abstract group elements $P_0,P_1$. Therefore we can only use the state-dependent relations derived in the previous section. We must show that
\begin{align}\label{eq:giant}
&f_A(J^{i} P_0^{j} (P_0 P_1^{-1})^{q_1} (P_0^{-1} P_1)^{q_2}) f_A(J^{i'} P_0^{j'} (P_0 P_1^{-1})^{q'_1} (P_0^{-1} P_1)^{q'_2})\ket{\psi} \nonumber\\&\quad= f_A(J^{i} P_0^{j} (P_0 P_1^{-1})^{q_1} (P_0^{-1} P_1)^{q_2}    J^{i'} P_0^{j'} (P_0 P_1^{-1})^{q'_1} (P_0^{-1} P_1)^{q'_2})\ket{\psi}
\end{align}
for all $i,j,i',j' \in [3]$ and $q_1,q_2,q'_1,q'_2 \in [2]$.  
\begin{claim}
Without loss of generality, we can assume $i=j=i'=q'_1=q'_2=0$.
\end{claim}
\begin{proof}
Fix $i,j,q_1,q_2,i',j',q'_1,q'_2$. We first show that without loss of generality we can assume $q'_1=q'_2=0$. By Lemma \ref{normform}, there exist $i'',j''\in [3],q''_1,q''_2\in[2]$ such that
\begin{align*} 
\bigparen{J^{i} P_0^{j} (P_0 P_1^{-1})^{q_1} (P_0^{-1} P_1)^{q_2}} \bigparen{J^{i'} P_0^{j'}} = J^{i''} P_0^{j''} (P_0 P_1^{-1})^{q''_1} (P_0^{-1} P_1)^{q''_2}.
\end{align*}
So it also holds that
\begin{align*} 
\bigparen{J^{i} P_0^{j} (P_0 P_1^{-1})^{q_1} (P_0^{-1} P_1)^{q_2}} \bigparen{J^{i'} P_0^{j'} (P_0 P_1^{-1})^{q'_1} (P_0^{-1} P_1)^{q'_2}} = J^{i''} P_0^{j''} (P_0 P_1^{-1})^{q''_1 + q'_1} (P_0^{-1} P_1)^{q''_2+q'_2} 
\end{align*}
since by Lemma \ref{comm_lemma}, $P_0 P_1^{-1}$ and $P_0^{-1} P_1$ commute. So the right-hand-side of (\ref{eq:giant}) can be written
\begin{align*}
&f_A(J^{i} P_0^{j} (P_0 P_1^{-1})^{q_1} (P_0^{-1} P_1)^{q_2}    J^{i'} P_0^{j'} (P_0 P_1^{-1})^{q'_1} (P_0^{-1} P_1)^{q'_2})\ket{\psi}\\&\quad=
f_A(J^{i''} P_0^{j''} (P_0 P_1^{-1})^{q''_1 + q'_1} (P_0^{-1} P_1)^{q''_2+q'_2} )\ket{\psi}\\
&\quad=
\omega^{i''} A_0^{j''} (A_0 A_1^{-1})^{q''_1 + q'_1} (A_0^{-1} A_1)^{q''_2+q'_2} \ket{\psi}\\
&\quad=(B_0 B_1)^{q'_2} (B_0^\ast B_1^\ast)^{q'_1}
\omega^{i''} A_0^{j''} (A_0 A_1^{-1})^{q''_1} (A_0^{-1} A_1)^{q''_2} \ket{\psi}\\
&\quad=(B_0 B_1)^{q'_2} (B_0^\ast B_1^\ast)^{q'_1}
f_A(J^{i''} P_0^{j''} (P_0 P_1^{-1})^{q''_1} (P_0^{-1} P_1)^{q''_2}) \ket{\psi}\\
&\quad=(B_0 B_1)^{q'_2} (B_0^\ast B_1^\ast)^{q'_1}
f_A\paren{\paren{J^{i} P_0^{j} (P_0 P_1^{-1})^{q_1} (P_0^{-1} P_1)^{q_2}} \paren{J^{i'} P_0^{j'}}} \ket{\psi},
\end{align*}
where in the fourth equality, we used (\ref{A_0A_1_3}) and (\ref{A_0A_1_4}) and the fact that Alice and Bob's operators commute.

Also since Alice and Bob's operators commute 
\begin{align*}
f_A(J^{i'} P_0^{j'} (P_0 P_1^{-1})^{q'_1} (P_0^{-1} P_1)^{q'_2})\ket{\psi} &= \omega^{i'} A_0^{j'} (A_0 A_1^\ast)^{q'_1} (A_0^\ast A_1)^{q'_2} \ket{\psi}\\
&=(B_0 B_1)^{q'_2}\omega^{i'} A_0^{j'} (A_0 A_1^\ast)^{q'_1}\ket{\psi}\\
&=(B_0 B_1)^{q'_2} (B_0^\ast B_1^\ast)^{q'_1} \omega^{i'} A_0^{j'} \ket{\psi}\\
&=(B_0 B_1)^{q'_2} (B_0^\ast B_1^\ast)^{q'_1} f_A(J^{i'} P_0^{j'})\ket{\psi}.
\end{align*}
Therefore the left-hand-side of (\ref{eq:giant}) can be written as
\begin{align*}
&f_A(J^{i} P_0^{j} (P_0 P_1^{-1})^{q_1} (P_0^{-1} P_1)^{q_2}) f_A(J^{i'} P_0^{j'} (P_0 P_1^{-1})^{q'_1} (P_0^{-1} P_1)^{q'_2})\ket{\psi} \\&\quad= (B_0 B_1)^{q'_2} (B_0^\ast B_1^\ast)^{q'_1} f_A(J^{i} P_0^{j} (P_0 P_1^{-1})^{q_1} (P_0^{-1}P_1)^{q_2})f_A(J^{i'} P_0^{j'})\ket{\psi}
\end{align*}

Since $B_0,B_1$ are unitaries, (\ref{eq:giant}) is equivalent to the following identity
\begin{align*}
f_A(J^{i} P_0^{j} (P_0 P_1^{-1})^{q_1} (P_0^{-1}P_1)^{q_2})f_A(J^{i'} P_0^{j'})\ket{\psi} = f_A\paren{\paren{J^{i} P_0^{j} (P_0 P_1^{-1})^{q_1} (P_0^{-1} P_1)^{q_2}} \paren{J^{i'} P_0^{j'}}} \ket{\psi},    
\end{align*}
in other words we can assume without loss of generality $q'_1=q'_2=0$. The case of $i=j=0$ is handled similarly. Also since $J$ and $f(J)$ are both central, we can assume $i' = 0$. 
\end{proof}
By this claim, we just need to verify
\begin{align}
f_A((P_0 P_1^{-1})^{q_1} (P_0^{-1} P_1)^{q_2}) f_A(P_0^{j'})\ket{\psi} = f_A((P_0 P_1^{-1})^{q_1} (P_0^{-1} P_1)^{q_2}P_0^{j'})\ket{\psi}
\end{align}
There are 12 cases to consider: $q_1,q_2 \in [2], j'\in [3]$. The case of $j'=0$ is trivial, and the case of $j'=2$ is handled similar to the case of $j'=1$. So we only consider the case of $j'=1$. The case of $q_1=q_2 = 0$ is trivial. We analyse the remaining three cases one-by-one:
\begin{itemize}
\item $q_1=0,q_2 =1$: First note that
\begin{align*}
(P_0^{-1} P_1) P_0 = P_0 P_0 P_1^{-1} P_1^{-1} P_0 = J^2 P_0 (P_0 P_1^{-1}) (P_0^{-1} P_1),
\end{align*}
which allows us to write
\begin{align*}
f_A((P_0^{-1} P_1)) f_A(P_0) \ket{\psi}&=  A_0^\ast A_1 A_0 \ket{\psi} \\
&= A_0^\ast A_1^\ast A_1^\ast A_0 \ket{\psi}\\
&= \omega^\ast A_0^\ast A_1^\ast A_0^\ast A_1 \ket{\psi}\\
&= \omega^\ast A_0 (A_0 A_1^\ast) (A_0^\ast A_1)\ket{\psi}\\
&= f_A(J^2 P_0 (P_0 P_1^{-1}) (P_0^{-1} P_1)) \ket{\psi} \\
&= f_A((P_0^{-1} P_1) P_0)\ket{\psi},
\end{align*}
where in the third line we used (\ref{group_rel_1}). 

\item $q_1=1,q_2=0$:
\begin{align*}
(P_0 P_1^{-1}) P_0 = J^2 P_0 (P_0^{-1} P_1)
\end{align*}
which allows us to write
\begin{align*}
f_A(P_0 P_1^{-1}) f_A(P_0)\ket{\psi} &= (A_0 A_1^\ast) A_0 \ket{\psi} \\
&= A_0 (A_1^\ast A_0)\ket{\psi} \\
&= \omega^\ast A_0 (A_0^\ast A_1)\ket{\psi} \\
&= f_A(J^{2} P_0 (P_0^{-1} P_1))\ket{\psi} \\
&= f_A((P_0 P_1^{-1}) P_0)\ket{\psi},
\end{align*}
where in the third line we used (\ref{group_rel_1}). 

\item $q_1=q_2=1$:
\begin{align*}
(P_0 P_1^{-1}) (P_0^{-1} P_1) P_0 = J (P_0 P_1^{-1}) (P_1^{-1} P_0) P_0 = J P_0 (P_1 P_0^{-1}) = J^{2} P_0 (P_0 P_1^{-1}).
\end{align*}
Now write
\begin{align*}
f_A((P_0 P_1^{-1}) (P_0^{-1} P_1)) f_A(P_0)\ket{\psi} &= A_0 A_1^\ast A_0^\ast A_1 A_0\ket{\psi}\\
&=  A_0 A_1^\ast A_0 A_0 A_1 A_0\ket{\psi}\\
&= \omega A_0 A_1^\ast A_0 A_0^\ast A_1^\ast A_0^\ast\ket{\psi}\\
&= \omega A_0 (A_1 A_0^\ast)\ket{\psi}\\
&= \omega^\ast A_0 (A_0 A_1^\ast)\ket{\psi}\\
&= f_A(J^{2} P_0 (P_0 P_1^{-1}))\ket{\psi} \\
&= f_A((P_0 P_1^{-1}) (P_0^{-1} P_1) P_0)\ket{\psi},
\end{align*}
where in the third line we used Proposition \ref{eq:important_1} and in the second last line we used (\ref{group_rel_2}). 

\end{itemize}
The proof that $f_B$ is a $\ket{\psi}$-representation follows similarly. 
\end{proof}

\begin{thm}
$\mathcal{G}_3$ is rigid.
\end{thm} 
\begin{proof}
The representation theory of $G_3$ is simple. There are nine irreducible representation of dimension one: These are given by $P_0\mapsto \omega^i, P_1\mapsto \omega^j, J\mapsto \omega^{2(j-i)}$ for $i,j\in [3]$. It also has three irreducible representations $g_1,g_2,g_3$ of dimension three defined by
\begin{align*}
g_1(P_0) &= \begin{pmatrix} 0 &0 &1 \\ 1 &0 &0\\ 0 &1 &0 \end{pmatrix}, \ g_1(P_1) = \begin{pmatrix} 0 &0 &\omega^\ast\\ -\omega^\ast &0 &0\\ 0 &-\omega^\ast &0\end{pmatrix}, \
  g_1(J) =\begin{pmatrix} \omega &0 &0\\ 0 &\omega &0\\0 &0 &\omega\end{pmatrix},\\
g_2(P_0) &= \begin{pmatrix}0 &0 &1\\1 &0 &0\\ 0 &1 &0\end{pmatrix}, \ g_2(P_1) = \begin{pmatrix}0 &0 &-1\\ -1 &0 &0\\0 &1 &0\end{pmatrix}, \ g_2(J) = \begin{pmatrix}1 &0 &0\\ 0 &1 &0\\0 &0 &1\end{pmatrix},\\
g_3(P_0) &= \begin{pmatrix} 0 &1 &0\\ 0 &0 &1\\1 &0 &0\end{pmatrix}, \ g_3(P_1) = \begin{pmatrix}0 &\omega &0\\ 0 &0 &-\omega\\ -\omega &0 &0\end{pmatrix}, \ 
g_3(J) = \begin{pmatrix}\omega^\ast &0 &0\\ 0 &\omega^\ast &0\\ 0 &0 &\omega^\ast\end{pmatrix}.
\end{align*}

Among these $g_1$, is the only representation that gives rise to an optimal strategy. This follows from a simple enumeration of these 12 irreducible representations. However we could also immediately see this, since $g_1$ is the only irreducible representation that satisfies the ring relation $H_3+I=0$.

Define a unitarily equivalent irreducible representation $g'_1 = U g_1 U^\ast$ where $U = \begin{pmatrix}0 &1 &0\\1 &0 &0\\0 &0 &1 \end{pmatrix}$. Now $\tilde{A}_0 = g_1(P_0), \tilde{A}_1 = g_1(P_1), \tilde{B}_0 = g'_1(P_0)^\ast, \tilde{B}_1\coloneqq g'_1(P_1)$ is the same strategy defined in example \ref{z3_example}. 

In addition $$\ket{\psi_3} = \frac{1}{\sqrt{10}} \Paren{ (1 - z^{4})\ket{00} + 2\ket{12} + (1 + z^{2})\ket{21} }$$ is the unique state that maximizes $\nu(\mathcal{G}_3,\mathcal{S}_{g_1,g'_1,\ket{\psi}})$. This follows since $\ket{\psi_3}$ is the unique eigenvector associated with the largest eigenvalue of $\mathcal{B}_3(\tilde{A}_0,\tilde{A}_1,\tilde{B}_0,\tilde{B}_1)$. The rigidity of $\mathcal{G}_3$ follows from Corollary \ref{corollary: application of gh}.

\end{proof}

\begin{rmk} The game $\mathcal{G}_3$ is in fact a robust self-test. We omit the proof, but at a high-level, if a strategy $(\{A_0,A_1\},\{B_0,B_1\},\ket{\psi})$ is $\epsilon$-optimal for $\mathcal{G}_3$, then
\[\bra{\psi}(6I - \mathcal{B}_3)\ket{\psi}\leq O(\epsilon).\]
Consequently, $\norm{S_i\ket{\psi}}\leq O(\sqrt{\epsilon}), \norm{T_j\ket{\psi}}\leq O(\sqrt{\epsilon})$ for all $i \in [2],j \in [6]$. From which one obtains a robust version of every relation in this section. 
\end{rmk}

\section{SOS approach to solution group}\label{BCSwithSOS}
In this section we show that the connection between an LCS game over $\mathbb{Z}_2$ and its solution group shown in \cite{BCSCommuting} can be determined using sum of squares techniques. 

We will suppress the tensor product notation and simply represent a strategy for an LCS game $\mathcal{G}_{A,b}$ by a state $\ket{\psi} \in \mathcal{H}$ and a collection of commuting measurement systems $\lbrace E_{i,x} \rbrace $ and $ \lbrace F_{j,y} \rbrace $. Using the notation outlined in section \ref{LCS Prelims} we define the following sets of observables
\begin{itemize}
    \item{Alice's Observables:} $  A_{j}^{(i)}= \sum_{x: x_j=1}E_{i,x} - \sum_{x:x_j=-1}E_{i,x}$, for each $i \in [r]$ and $j\in V_i$
    \item{Bob's Observables:} $B_j=F_{j,1} - F_{j,-1}$ for each $j \in [s]$. 
\end{itemize}
Note $A_{j}^{(i)}$ commutes with $A_{j'}^{(i)}$ for all $i \in [r]$ and $j,j' \in V_i$ and $B_j$ commutes with $A_{j}^{(i)}$ for all $i,j$. These observables will satisfy the following identities: 

\begin{align}
    \sum_{x: x \in S_{i}} E_{i,x} &= \frac{1}{2} \left( I+ (-1)^{b_{i}}\prod_{k \in V_{i}} A^{(i)}_{k} \right) \label{eq1} \\
    \sum_{x: y=x_{j}} E_{i,x} &= \frac{1}{2} \left( I+yA_{j}^{(i)} \right) \label{eq2}
\end{align}

The probability of Alice and Bob winning the game is given by evaluating $\bra{\psi} v \ket{\psi} $ where
\begin{align*} v &= \sum_{\substack{i \in [r] \\j \in V_i}} \frac{1}{r |V_i|} \left( \sum_{\substack{x ,y: \\ x \in S_i \\ y=x_j }} E_{i,x}F_{j,y} \right) \\ &=\sum_{i,j} \frac{1}{2 r |V_i|}\left(1- \sum_{\substack{x ,y: \\ x \in S_i \\ y=x_j }} E_{i,x}F_{j,y} \right)^2.\end{align*}

Observe using identities \ref{eq1} and \ref{eq2} we have    

\begin{align*}
\left( 1- \sum_{\substack{x ,y: \\ x \in S_i \\ y=x_j }} E_{i,x}F_{j,y} \right) &=I-\sum_y F_{j,y}\sum_{\substack{x : \\ x \in S_i \\ y=x_j }}E_{i,x} \\
    &=I-\frac{1}{4}\sum_y F_{j,y} \left((I+yA_{j}^{(i)})(I+(-1)^{b_{i}}\prod_{k \in v_{i}} A_{k}^{(i)}) \right) \\ &= I-\frac{1}{4}\sum_y F_{j,y}\Bigg(I+yA_{j}^{(i)}+(-1)^{b_{i}}\prod_{k \in v_{i}} A_{k}^{(i)}+y(-1)^{b_{i}}\prod_{k \in v_{i}} A_{k}^{(i)} A_{j}^{(i)}\Bigg)\\
    &=I-\frac{1}{4}F_{j,1}\Bigg(I+A_{j}^{(i)}\!\begin{aligned}[t]&+(-1)^{b_{i}}\prod_{k \in v_{i}} A_{k}^{(i)}
    +(-1)^{b_{i}}\prod_{k \in v_{i}} A_{k}^{(i)}A_{j}^{(i)}\Bigg)\\
    &- \frac{1}{4}F_{j,-1} \Bigg(I -A_{j}^{(i)}+(-1)^{b_{i}}\prod_{k \in v_{i}} A_{k}^{(i)}+-(-1)^{b_{i}}\prod_{k \in v_{i}} A_{k}^{(i)} A_{j}^{(i)}\Bigg)\end{aligned}\\
    \phantom{R_{i,j}} &= I - \frac{1}{4}I -\frac{1}{4}B_jA_{j}^{(i)} -\frac{1}{4}(-1)^{b_{i}}\prod_{k \in v_{i}} A_{k}^{(i)} -\frac{1}{4}B_j (-1)^{b_{i}}\prod_{k \in v_{i}} A_{k}^{(i)}A_{j}^{(i)}\\
   &= \frac{1}{8}\Bigg( (I-B_jA_{j}^{(i)})^2 + (I-(-1)^{b_{i}}\prod_{k \in v_{i}} A_{k}^{(i)})^2 + (I- (-1)^{b_{i}}\prod_{k \in v_{i}} A_{k}^{(i)}A_{j}^{(i)}B_j)^2 \Bigg).
\end{align*}

Thus Alice and Bob are using a perfect strategy if and only if \[
0 =(I-B_jA_{j}^{(i)}) \ket{\psi} \\ = (I-(-1)^{b_{i}}\prod_{k \in v_{i}}A_{k}^{(i)}) \ket{\psi}\\ = (I- (-1)^{b_{i}}\prod_{k \in v_{i}} A_{k}^{(i)}A_{j}^{(i)}B_j) \ket{\psi}. \]
The above equalities will hold exactly when the following two identities hold  for all  $i$  and $j \in V_i$,
\begin{align}
     B_j \ket{\psi}&=A_{j}^{(i)} \ket{\psi} \label{eq3} \\
      \ket{\psi} &=  (-1)^{b_{i}} \prod_{k \in v_{i}}A_{k}^{(i)} \ket{\psi} \label{eq4}
\end{align}
Using identities \ref{eq3} and \ref{eq4} it is possible to define a $\ket{\psi}$-representation for the solution group $G_{A,b}$.

\section{A non-rigid pseudo-telepathic LCS game} \label{semirigid}

The canonical example of a pseudo-telepathic LCS games is the Mermin-Peres magic square game \cite{Mermin} defined in the following figure.

\begin{figure}[H]
\[\begin{array}{c@{}c@{}c@{}c@{}c@{}c}
e_{1} & \text{ --- } & e_{2} & \text{ --- } & e_{3} \\
 \vert & & \vert & & \vert\vert \\
e_{4} & \text{ --- } & e_{5} & \text{ --- } & e_{6} \\
 \vert & & \vert & & \vert\vert\\
e_{7} & \text{ --- } & e_{8} & \text{ --- } & e_{9} \\
\end{array}\]

\caption{This describes the Mermin-Peres magic square game. Each single-line indicates that the variables along the line multiply to $1$, and the double-line indicates that the variables along the line multiply to $-1$. }
\label{fig:magic_square}

\end{figure}

It is well-known that the Mermin-Peres magic square game has the following operator solution for which the corresponding quantum strategy is rigid \cite{Wu}.

\begin{gather*}
    A_{1} = I \otimes \sigma_{Z}, \quad A_{2} = \sigma_{Z} \otimes I, \quad A_{3} = \sigma_{Z} \otimes \sigma_{Z} \\
    A_{4} = \sigma_{X} \otimes I, \quad A_{5} = I \otimes \sigma_{X}, \quad A_{6} = \sigma_{X} \otimes \sigma_{X} \\
    A_{7} = \sigma_{X} \otimes \sigma_{Z}, \quad A_{8} = \sigma_{Z} \otimes \sigma_{X}, \quad A_{9} = \sigma_{Y} \otimes \sigma_{Y},
\end{gather*}

In this section, we provide an example of a non-local game whose perfect solutions must obey particular group relations but is not a self-test. This game, \emph{glued magic square}, is described in Figure \ref{fig:2_magic_square}.

\begin{figure}[H]
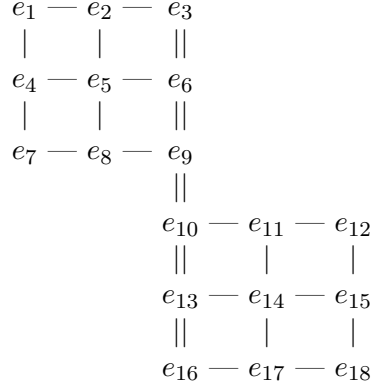

\[\begin{array}{c@{}c@{}c@{}c@{}c@{}c@{}c@{}c@{}c}
e_{1} & \text{ --- } & e_{2} & \text{ --- } & e_{3} & & & & \\
 \vert & & \vert & & \vert\vert& & & & \\
e_{4} & \text{ --- } & e_{5} & \text{ --- } & e_{6} & & & &\\
 \vert & & \vert & & \vert\vert & & & &\\
e_{7} & \text{ --- } & e_{8} & \text{ --- } & e_{9} & & & & \\

& & & & \vert\vert & & & \\

& & & & e_{10} & \text{ --- } & e_{11} & \text{ --- } & e_{12}\\
& & & & \vert\vert & & \vert & & \vert \\
& & & & e_{13} & \text{ --- } & e_{14} & \text{ --- } & e_{15}\\
& & & & \vert\vert & & \vert & & \vert \\
& & & & e_{16} & \text{ --- } & e_{17} & \text{ --- } & e_{18}\\
\end{array}\]

\caption{This describes a LCS game with 18 variables $e_{1}, e_{2}, \dots, e_{18}$. Each single-line indicates that the variables along the line multiply to $1$, and the double-line indicates that the variables along the line multiply to $-1$. }
\label{fig:2_magic_square}

\end{figure}

In order to show that this game is not a self-test, we first define two operator solutions, that give rise to perfect strategies. Let $\mathcal{E} = \{E_{1}, E_{2}, \dots, E_{18}\}$ be defined as
\[ E_i = \begin{cases}
\begin{pmatrix}
I_4 & 0 \\ 0 & A_i 
\end{pmatrix} & \text{ for } i = 1, 2, \dots, 9 \\ \begin{pmatrix}
A_{i-9} & 0 \\ 0 & I_4
\end{pmatrix} & \text { for } i= 10, 11, \dots, 18
\end{cases}\]
and $\mathcal{F} = \{F_{1}, F_{2}, \dots, F_{18}\}$ as
\[
F_i = \begin{cases}
A_{i} & \text{ for } i = 1, 2 \dots, 9 \\ 
I_4  & \text { for } i= 10, 11 \dots, 18
\end{cases}\]

These two operators solutions $\mathcal{E}$ and $\mathcal{F}$ give rise to two quantum strategies with the entangled states $\ket{\psi_{1}} = \frac{1}{\sqrt{8}} \sum_{i=0}^{7} \ket{i}\ket{i}$ and $\ket{\psi_{2}} = \frac{1}{2} \sum_{i=0}^{3} \ket{i}\ket{i}$.

\begin{thm}
The glued magic square game is not a self-test for any quantum strategy.
\end{thm}
\begin{proof}
Suppose, for the sake of contradiction, there is a quantum strategy $\left(\lbrace A_i \rbrace_i, \lbrace B_j \rbrace_j \ket{\psi} \right) $ that is rigid. Then there exist local isometries $U_{A}$, $U_{B}$ and $V_{A}$, $V_{B}$ such that
\begin{gather}
    (U_{A}E_{1} \otimes U_{B})\ket{\psi_{1}} = ((A_{1} \otimes I)\ket{\psi})\ket{\text{junk}_{1}} \label{Msquare_rel_1} \\
    (U_{A}E_{5} \otimes U_{B})\ket{\psi_{1}} = ((A_{5} \otimes I)\ket{\psi})\ket{\text{junk}_{1}} \label{Msquare_rel_2} \\
    (V_{A}F_{1} \otimes V_{B})\ket{\psi_{2}} = ((A_{1} \otimes I)\ket{\psi})\ket{\text{junk}_{2}} \label{Msquare_rel_3} \\
    (V_{A}F_{5} \otimes V_{B})\ket{\psi_{2}} = ((A_{5} \otimes I)\ket{\psi})\ket{\text{junk}_{2}}. \label{Msquare_rel_4}
\end{gather}

From relation (\ref{Msquare_rel_2}), we obtain 
\[ \bra{\psi_{1}}(E_{5}U_{A}^{*} \otimes U_{B}^{*}) = \bra{\text{junk}_{1}}(\bra{\psi}(A_{5}^{*} \otimes I)), \]
and hence together with relation (\ref{Msquare_rel_1}), we obtain the following relation between $E_{5}E_{1}$ and $A_{5}^{*}A_{1}$
\[ \bra{\psi_{1}}(E_{5}E_{1} \otimes I)\ket{\psi_{1}} = \bra{\psi}(A_{5}^{*}A_{1} \otimes I)\ket{\psi}. \]
Similarly, we also obtain
\[ \bra{\psi_{2}}(F_{5}F_{1} \otimes I)\ket{\psi_{2}} = \bra{\psi}(A_{5}^{*}A_{1} \otimes I)\ket{\psi}, \]
and hence
\[ \bra{\psi_{1}}(E_{5}E_{1} \otimes I)\ket{\psi_{1}} = \bra{\psi_{2}}(F_{5}F_{1} \otimes I)\ket{\psi_{2}}. \]

By first applying the adjoint to relation (\ref{Msquare_rel_1}) and (\ref{Msquare_rel_3}), we obtain
\[ \bra{\psi_{1}}(E_{1}E_{5} \otimes I)\ket{\psi_{1}} = \bra{\psi_{2}}(F_{1}F_{5} \otimes I)\ket{\psi_{2}}. \]
Now, since $F_{1}$ and $F_{5}$ anti-commute, we get the following relation between $E_{5}E_{1}$ and $E_{1}E_{5}$
\[ \bra{\psi_{1}}(E_{5}E_{1} \otimes I)\ket{\psi_{1}} = -\bra{\psi_{1}}(E_{1}E_{5} \otimes I)\ket{\psi_{1}}. \]
However, a direct computation of $\bra{\psi_{1}}(E_{5}E_{1} \otimes I)\ket{\psi_{1}}$ shows that
\[ \bra{\psi_{1}}(E_{5}E_{1} \otimes I)\ket{\psi_{1}} = \frac{1}{8} \sum_{i = 0}^{7} \bra{i}E_{5}E_{1}\ket{i} = \frac{1}{8} \Tr(E_{5}E_{1}) = \frac{1}{8} \Tr(E_{1}E_{5}) = \bra{\psi_{1}}(E_{1}E_{5} \otimes I)\ket{\psi_{1}}, \]
and $\Tr(E_{1}E_{5}) = \Tr(I_{4}) + \Tr(I \otimes \sigma_{Z}\sigma_{X}) = 4 \neq 0$. Hence, the glued magic square game is not rigid. 
\end{proof}

Although this game is not a self-test, we know from Section \ref{BCSwithSOS} Alice's operators must provide a $\ket{\psi}$-representation for the solution group of glued magic square, and thus must satisfy particular group relations.

\bibliographystyle{alpha}
\bibliography{bibliography.bib}

\end{document}